\documentclass[11pt]{article}
\usepackage[left=.9in, right=.9in, top=.9in, bottom=.9in]{geometry}
\usepackage{amssymb}
\usepackage{amsmath} 
\usepackage{amsthm}
\usepackage{verbatim}
\usepackage{enumerate}
\usepackage{hyperref}
\hypersetup{
	colorlinks=true,
	citecolor=blue,
	linkcolor=blue,
	urlcolor=blue
}
\usepackage{cleveref}

\DeclareMathOperator{\GL}{GL}

\DeclareMathOperator{\U}{U}
\DeclareMathOperator{\T}{T}
\DeclareMathOperator{\SU}{SU}
\DeclareMathOperator{\SO}{SO}
\DeclareMathOperator{\PU}{PU}

\DeclareMathOperator{\Hom}{Hom}
\DeclareMathOperator{\Res}{Res}
\DeclareMathOperator{\Ind}{Ind}
\DeclareMathOperator{\Rep}{Rep}
\DeclareMathOperator{\Tr}{Tr}
\DeclareMathOperator{\Pker}{Pker}

\newcommand{\ket}[1]{\left|#1\right\rangle}

\newtheorem{theorem}{Theorem}[section]
\newtheorem{corollary}{Corollary}[theorem]
\newtheorem{lemma}[theorem]{Lemma}
\newtheorem{proposition}[theorem]{Proposition}

\theoremstyle{remark}
\newtheorem*{remark}{Remark}

\theoremstyle{definition}
\newtheorem{definition2}[theorem]{Definition}

\theoremstyle{remark}
\newtheorem{openProblem}{Open Problem}

\title{The Power of a Single Qubit: Two-way Quantum Finite Automata and the Word Problem}
\author{Zachary Remscrim \\ Department of Mathematics \\ MIT \\ remscrim@mit.edu}
\date{}
\begin{document}

\maketitle

\begin{abstract}
	The two-way finite automaton with quantum and classical states (2QCFA), defined by Ambainis and Watrous, is a model of quantum computation whose quantum part is extremely limited; however, as they showed, 2QCFA are surprisingly powerful: a 2QCFA, with a single qubit, can recognize, with bounded error, the language $L_{eq}=\{a^m b^m :m \in \mathbb{N}\}$ in expected polynomial time and the language $L_{pal}=\{w \in \{a,b\}^*:w \text{ is a palindrome}\}$ in expected exponential time.  
	
	We further demonstrate the power of 2QCFA by showing that they can recognize the word problems of many groups. In particular 2QCFA, with a single qubit and algebraic number transition amplitudes, can recognize, with bounded error, the word problem of any finitely generated virtually abelian group in expected polynomial time, as well as the word problems of a large class of linear groups in expected exponential time. This latter class (properly) includes all groups with context-free word problem. We also exhibit results for 2QCFA with any constant number of qubits.
	
	As a corollary, we obtain a direct improvement on the original Ambainis and Watrous result by showing that $L_{eq}$ can be recognized by a 2QCFA with better parameters. As a further corollary, we show that 2QCFA can recognize certain non-context-free languages in expected polynomial time.
	
	In a companion paper, we prove matching lower bounds, thereby showing that the class of languages recognizable with bounded error by a 2QCFA in expected \textit{subexponential} time is properly contained in the class of languages recognizable with bounded error by a 2QCFA in expected \textit{exponential} time.
	
\end{abstract}

\section{Introduction}\label{sec:intro}

The theory of quantum computation has made amazing strides in the last several decades. Landmark results, like Shor's polynomial time quantum algorithm for integer factorization \cite{shor1994algorithms}, Grover's algorithm for unstructured search \cite{grover1996fast}, and the linear system solver of Harrow, Hassidim, and Lloyd \cite{harrow2009quantum}, have provided remarkable examples of natural problems for which quantum computers seem to have an advantage over their classical counterparts. These theoretical breakthroughs have provided strong motivation to construct quantum computers. However, while significant advancements have been made, the experimental quantum computers that exist today are still quite limited, and are certainly not capable of implementing, on a large scale, algorithms designed for general quantum Turing machines. This naturally motivates the study of more restricted models of quantum computation.

In this paper, our goal is to understand the computational power of a small number of qubits, especially the power of a single qubit. To that end, we study two-way finite automata with quantum and classical states (2QCFA), introduced by Ambainis and Watrous \cite{ambainis2002two}. Informally, a 2QCFA is a two-way deterministic finite automaton (2DFA) that has been augmented with a quantum register of constant size, i.e., a constant number of qubits. The quantum part of the machine is extremely limited; however, the model is surprisingly powerful. In particular, Ambainis and Watrous \cite{ambainis2002two} showed that a 2QCFA, using only one qubit, can recognize, with bounded error, the language $L_{eq}=\{a^m b^m :m \in \mathbb{N}\}$ in expected polynomial time and the language $L_{pal}=\{w \in \{a,b\}^*:w \text{ is a palindrome}\}$ in expected exponential time. This clearly demonstrated that 2QCFA are more powerful than 2DFA, which recognize precisely the regular languages \cite{rabin1959finite}. Moreover, as it is known that two-way probabilistic finite automata (2PFA) can recognize $L_{eq}$ with bounded error in exponential time \cite{freivalds1981probabilistic}, but not in subexponential time \cite{greenberg1986lower}, and cannot recognize $L_{pal}$ with bounded error in any time bound \cite{dwork1992finite}, this result also demonstrated the superiority of 2QCFA over 2PFA. 

We investigate the ability of 2QCFA to recognize the word problem of a group. Informally, the word problem for a group $G$ involves determining if the product of a finite sequence of group elements $g_1,\ldots,g_k \in G$ is equal to the identity element of $G$. Word problems for various classes of groups have a rich and well-studied history in computational complexity theory, as there are many striking relationships between certain algebraic properties of a group $G$ and the computational complexity of its word problem $W_G$. For example, $W_G \in \mathsf{REG} \Leftrightarrow G$ is finite \cite{anisimov1971group}, $W_G \in \mathsf{CFL} \Leftrightarrow W_G \in \mathsf{DCFL} \Leftrightarrow G$ is a finitely generated virtually free group \cite{muller1983groups}, and $W_G \in \mathsf{NP} \Leftrightarrow G$ is a finitely generated subgroup of a finitely presented group with polynomial Dehn function \cite{birget2002isoperimetric}.

For a quantum model, such as the 2QCFA, word problems are a particularly natural class of languages to study. There are several results \cite{brodsky2002characterizations,yakaryilmaz2010succinctness,yakaryilmaz2010languages} which show that certain (generally significantly more powerful) QFA variants can recognize the word problems of particular classes of groups (see the excellent survey \cite{ambainis2015automata} for a full discussion of the many QFA variants). Moreover, there are also results concerning the ability of QFA to recognize certain languages that are extremely closely related to word problems; in fact, the languages $L_{eq}$ and $L_{pal}$ considered by Ambainis and Watrous \cite{ambainis2002two} are each closely related to a word problem.  

Fundamentally, the laws of quantum mechanics sharply constrain the manner in which the state of the quantum register of a 2QCFA may evolve, thereby forcing the computation of a 2QCFA to have a certain algebraic structure. Similarly, the algebraic properties of a particular group $G$ impose a corresponding algebraic structure on its word problem $W_G$. For certain classes of groups, the algebraic structure of $W_G$ is extremely compatible with the algebraic structure of the computation of a 2QCFA; for other classes of groups, these two algebraic structures are in extreme opposition.

In this paper, we show that there is a broad class of groups for which these algebraic structures are quite compatible, which enables us to produce 2QCFA that recognize these word problems. As a corollary, we show that $L_{eq}$ can be recognized by a 2QCFA with better parameters than in the original Ambainis and Watrous result \cite{ambainis2002two}. 

In a separate paper \cite{remscrim2019lower}, we establish matching lower bounds on the running time of a 2QCFA (and, more generally, a \textit{quantum Turing machine} that uses sublogarithmic space) that recognizes these word problems, thereby demonstrating the optimality of these results; this allows us to prove that the class of languages recognizable with bounded error by 2QCFA in expected \textit{subexponential} time is properly contained in the class of languages recognizable with bounded error by 2QCFA in expected \textit{exponential} time.

\subsection{Statement of the Main Results}\label{sec:intro:results}

We show that, for many groups $G$, the corresponding word problem $W_G$ is recognized by a 2QCFA with ``good" parameters. In order to state these results, we must make use of some terminology and notation concerning 2QCFA, the word problem of a group, and various classes of groups whose word problems are of complexity theoretic interest. A full description of the 2QCFA model can be found in Section~\ref{sec:prelim:quant}; the definition of the word problem, as well as additional group theory background, including the definitions of the various classes of groups discussed in this section, can be found in Section~\ref{sec:prelim:wordProblem}. The following definition establishes some useful notation that will allow us to succinctly describe the parameters of a 2QCFA. We use $\mathbb{R}_{>0}$ to denote the positive real numbers.

\begin{definition2}\label{def:2QCFAparams}
	For $T:\mathbb{N} \rightarrow \mathbb{N}$, $\epsilon \in \mathbb{R}_{>0}$, $d \in \mathbb{N}$, and $\mathbb{A} \subseteq \mathbb{C}$, the complexity class $\mathsf{coR2QCFA}(T,\epsilon,d,\mathbb{A})$ consists of all languages $L \subseteq \Sigma^*$ for which there is a 2QCFA $M$, which has $d$ quantum basis states and transition amplitudes in $\mathbb{A}$, such that, $\forall w \in \Sigma^*$, the following holds: $M$ runs in expected time $O(T(
	\lvert w \rvert))$, $\Pr[M \text{ accepts } w]+\Pr[M \text{ rejects } w]=1$, $w \in L \Rightarrow \Pr[M \text{ accepts } w]=1$, and $w \not \in L \Rightarrow \Pr[M \text{ rejects } w] \geq 1-\epsilon$.	
\end{definition2}

The focus on the transition amplitudes of a 2QCFA warrants a bit of additional justification, as while it is standard to limit the transition amplitudes of a Turing machine in this way, it is common for finite automata to be defined without any such limitation.  For many finite automata models, applying such a constraint would be superfluous; for example, the class of languages recognized with bounded error and in expected time $2^{n^{o(1)}}$ by a 2PFA with no restriction at all on its transition amplitudes is precisely the regular languages \cite{dwork1990time}. However, the power of the 2QCFA model is quite sensitive to the choice of transition amplitudes. A 2QCFA with non-computable transition amplitudes can recognize undecidable languages, with bounded error and in expected polynomial time \cite{say2017magic}; whereas, 2QCFA with transition amplitudes restricted to the algebraic numbers $\overline{\mathbb{Q}}$ can only recognize languages in $\mathsf{P} \cap \mathsf{L}^2$, even if permitted unbounded error and exponential time \cite{watrous2003complexity}. In particular, the algebraic numbers are arguably the ``standard'' choice for the permitted transition amplitudes of a quantum Turing machine (QTM). It is desirable for the definition of 2QCFA to be consistent with that of QTMs as such consistency makes it more likely that techniques developed for 2QCFA could be applied to QTMs. Therefore, $\overline{\mathbb{Q}}$ is the the natural choice for the permitted transition amplitudes of a 2QCFA, though we do also consider the impact of allowing transition amplitudes in the slightly broader class $\widetilde{\mathbb{C}}=\overline{\mathbb{Q}} \cup \{e^{\pi i r}:r \in (\overline{\mathbb{Q}} \cap \mathbb{R})\}$.

We begin with a simple motivating example. For a finite alphabet $\Sigma$, a symbol $\sigma \in \Sigma$, and a word $w \in \Sigma^*$, let $\#(w,\sigma)$ denote the number of appearances of  $\sigma$ in $w$. Then the word problem for the group $\mathbb{Z}$ (the integers, where the group operation is addition) is the language $W_{\mathbb{Z}}=\{w \in \{a,b\}^*:\#(w,a)=\#(w,b)\}$. This language is closely related to the language $L_{eq}=\{a^m b^m :m \in \mathbb{N}\}$; in particular, $L_{eq}=(a^* b^*) \cap W_{\mathbb{Z}}$. More generally, the word problem for the group $\mathbb{Z}^k$ (the direct product of $k$ copies of $\mathbb{Z}$) is the language $W_{\mathbb{Z}^k}=\{w \in \{a_1,b_1,\ldots,a_k,b_k\}^*:\#(w,a_i)=\#(w,b_i), \forall i\}$.

Ambainis and Watrous \cite{ambainis2002two} showed that $L_{eq} \in \mathsf{coR2QCFA}(n^4,\epsilon,2,\widetilde{\mathbb{C}})$, $\forall \epsilon \in \mathbb{R}_{>0}$. We note that the same method would easily imply the same result for $W_{\mathbb{Z}}$, and could be further adapted to produce a similar result for $W_{\mathbb{Z}^k}$. Our first main theorem generalizes and improves upon these results in several ways. Let $\widehat{\Pi}_1$ denote the collections of all finitely generated virtually abelian groups (i.e., all groups that have a finite-index subgroup isomorphic to $\mathbb{Z}^k$, for some $k \in \mathbb{N}$, where $\mathbb{Z}^0$ is the trivial group); we will explain this choice of notation shortly. 

\begin{theorem}\label{thm:main:abelian}
	$\exists C \in \mathbb{R}_{>0},\forall G \in \widehat{\Pi}_1,\forall \epsilon \in \mathbb{R}_{>0}$, $W_G \in  (\mathsf{coR2QCFA}(n^3,\epsilon,2,\widetilde{\mathbb{C}}) \cap \mathsf{coR2QCFA}(n^C,\epsilon,2,\overline{\mathbb{Q}}))$.
\end{theorem}

By the above observation that $L_{eq}=(a^* b^*) \cap W_{\mathbb{Z}}$, the following corollary is immediate. 

\begin{corollary}
	$\exists C \in \mathbb{R}_{>0},\forall \epsilon \in \mathbb{R}_{>0}$, $L_{eq}\in  (\mathsf{coR2QCFA}(n^3,\epsilon,2,\widetilde{\mathbb{C}}) \cap \mathsf{coR2QCFA}(n^C,\epsilon,2,\overline{\mathbb{Q}}))$. 
\end{corollary}

The above corollary improves upon the result of Ambainis and Watrous \cite{ambainis2002two} in two distinct senses. Firstly, using the same set of permissible transition amplitudes, our result has a better expected running time. Secondly, our result shows that $L_{eq}$ can be recognized by a 2QCFA with transition amplitudes in $\overline{\mathbb{Q}}$, which still runs in expected polynomial time.

Let $\mathsf{CFL}$ denote the context-free languages (languages recognized by non-deterministic pushdown automata), $\mathsf{OCL}$ denote the one-counter languages (languages recognized by non-deterministic pushdown automata with single-symbol stack alphabet) and $\mathsf{poly{-}CFL}$ (resp. $\mathsf{poly{-}OCL}$) denote the intersection of finitely many context-free (resp. one-counter) languages. As $W_G \in \mathsf{poly{-}OCL} \Leftrightarrow G  \in \widehat{\Pi}_1$ \cite{holt2008groups}, the following corollary is also immediate. 

\begin{corollary}
	$\exists C \in \mathbb{R}_{>0},\forall W_G \in \mathsf{poly{-}OCL},\forall \epsilon \in \mathbb{R}_{>0}$, $W_G \in  (\mathsf{coR2QCFA}(n^3,\epsilon,2,\widetilde{\mathbb{C}})\cap$ \newline $\mathsf{coR2QCFA}(n^C,\epsilon,2,\overline{\mathbb{Q}}))$.
\end{corollary}

Moreover, as $W_G \in \mathsf{poly{-}OCL} \cap \mathsf{CFL} \Leftrightarrow G$ is a finitely generated virtually cyclic group \cite{holt2008groups}, the above corollary exhibits a wide class of non-context-free languages that are recognizable by a 2QCFA in polynomial time: the word problem of any group that is virtually $\mathbb{Z}^k$, $k \geq 2$.

\begin{remark}
	Interestingly, the limiting factor on the running time of the 2QCFA for any of the above word problems (or $L_{eq}$) is not the difficulty of distinguishing strings in the language from strings not in the language, but is instead due to the apparent difficulty of using a 2QCFA to produce a Boolean random variable with a particular (rather extreme) bias. In particular, we make use of the procedure (from \cite{ambainis2002two}) that allows a 2QCFA, on an input of size $n$, to generate a Boolean value that is $1$ with probability essentially $n^{-1}$, in time $O(n^2)$. If, for some $\delta \in (0,1)$, it were possible for a 2QCFA to produce a Boolean variable that has value $1$ with probability $n^{-\delta}$ in time $r(n)$, then our technique would immediately show that $\forall G \in \widehat{\Pi}_1,\forall \epsilon \in \mathbb{R}_{>0}$, $W_G \in \mathsf{coR2QCFA}((n+r(n))n^{\delta},\epsilon,2,\widetilde{\mathbb{C}})$.
\end{remark}

Next, let $F_k$ denote the free group of rank $k$, for $k \in \mathbb{N}$; in particular, $F_0$ is the trivial group, $F_1$ is the group $\mathbb{Z}$, and, for any $k \geq 2$, $F_k$ is non-abelian. Notice that $W_{F_2}$ is closely related to the language $L_{pal}$. Ambainis and Watrous \cite{ambainis2002two} showed that, $\forall \epsilon \in \mathbb{R}_{>0}$, $\exists D \in \mathbb{R}_{\geq 1}$, such that $L_{pal} \in \mathsf{coR2QCFA}(D^n,\epsilon,2,\overline{\mathbb{Q}})$, and the same method would show the same result for $W_{F_2}$. We show that the same result holds for any group built from free groups, using certain operations. Let $\widehat{\Pi}_2$ denote the collection of all finitely generated groups that are virtually a subgroup of a direct product of finitely many finite-rank free groups.

\begin{theorem}\label{thm:main:freeDirProd}
	$\forall G \in \widehat{\Pi}_2, \forall \epsilon \in \mathbb{R}_{>0},\exists D \in \mathbb{R}_{\geq 1}$, such that $W_G \in \mathsf{coR2QCFA}(D^n,\epsilon,2,\overline{\mathbb{Q}})$.
\end{theorem}

As $W_G \in \mathsf{CFL}\Leftrightarrow G$ is a finitely generated virtually free group \cite{muller1983groups}, we obtain the following.

\begin{corollary}
	$\forall W_G \in \mathsf{CFL}, \forall \epsilon \in \mathbb{R}_{>0},\exists D \in \mathbb{R}_{\geq 1}$, such that $W_G \in \mathsf{coR2QCFA}(D^n,\epsilon,2,\overline{\mathbb{Q}})$.
\end{corollary}

Consider the homomorphism $\pi:F_2 \times F_2 \rightarrow \mathbb{Z}$, where $\pi$ takes each free generator of each copy of $F_2$ to a single generator of $\mathbb{Z}$; then $K=\ker \pi$ is finitely generated, but not finitely presented \cite{stallings1963finitely}. All groups $G$ for which $W_G \in \mathsf{CFL}\cup \mathsf{poly{-}OCL}$ are finitely presented \cite{herbst1991subclass}. As $K \in \widehat{\Pi}_2$, we have the following corollary.

\begin{corollary}
	There is a finitely generated group $K$, which is not finitely presented (hence, $W_K \not \in \mathsf{CFL}\cup \mathsf{poly{-}OCL}$), where $\forall \epsilon \in \mathbb{R}_{>0},\exists D \in \mathbb{R}_{\geq 1}$, such that $W_K \in \mathsf{coR2QCFA}(D^n,\epsilon,2,\overline{\mathbb{Q}})$.
\end{corollary}

\begin{remark}
	It is known that, if $G \in \widehat{\Pi}_2$, then $W_G \in \mathsf{poly{-}CFL}$ \cite{brough2014groups}. Moreover, it is conjectured that $\widehat{\Pi}_2$ is precisely the class of groups whose word problem is in $\mathsf{poly{-}CFL}$ \cite{brough2014groups} (cf. \cite{ceccherini2015multipass}).
\end{remark}

We next consider a broader class of groups. Let $Z(H)$ denote the center of a group $H$, let $\U(d,\overline{\mathbb{Q}})$ denote the group of $d \times d$ unitary matrices with entries in $\overline{\mathbb{Q}}$, let $\PU(d,\overline{\mathbb{Q}})=\U(d,\overline{\mathbb{Q}})/Z(\U(d,\overline{\mathbb{Q}}))$, and let $(\PU(d,\overline{\mathbb{Q}}))^k$ denote the direct product of $k$ copies of $\PU(d,\overline{\mathbb{Q}})$. 

\begin{theorem}\label{thm:main:projEmbed}
	If $G$ is a finitely generated group that is virtually a subgroup of $(\PU(d,\overline{\mathbb{Q}}))^k$, for some $d,k \in \mathbb{N}_{\geq 1}$, then $\forall \epsilon \in \mathbb{R}_{>0},\exists D \in \mathbb{R}_{\geq 1}$, such that $W_G \in \mathsf{coR2QCFA}(D^n,\epsilon,d,\overline{\mathbb{Q}})$.
\end{theorem}

In order to state our final main result, as well as to provide appropriate context for the results listed above, we define the classes of groups $\Sigma_j$ and $\Pi_j$, for $j \in \mathbb{N}$, inductively. First $\Sigma_0=\Pi_0=\{\mathbb{Z},\{1\}\}$ (i.e., both classes consist of the two groups $\mathbb{Z}$ and the trivial group $\{1\}$). We use $\times$ to denote the direct product and $\ast$ to denote the free product. For $j >1$, we define $\Pi_j=\{ H_1 \times \cdots \times H_t:t \in \mathbb{N}_{\geq 1}, H_1,\ldots,H_t \in \Sigma_{j-1}\}$ and $\Sigma_j=\{H_1 \ast \cdots \ast H_t:t \in \mathbb{N}_{\geq 1}, H_1,\ldots,H_t \in \Pi_{j-1}\}$. These groups comprise an important subclass of a particularly important class of groups: the right-angled Artin groups. Note that every $G \in \bigcup_j (\Pi_j \cup \Sigma_j)$ is finitely generated. Also note that the $\Pi_j$ and $\Sigma_j$ form a hierarchy in the obvious way. We further define $\widehat{\Pi}_j$ (resp. $\widehat{\Sigma}_j$) as the set of all finitely generated groups that are virtually a subgroup of some group in $\Pi_j$ (resp. $\Sigma_j$), which also form a hierarchy in the obvious way. 

In particular, $\widehat{\Pi}_1$ (resp. $\widehat{\Pi}_2$) is precisely the class of groups for which \Cref{thm:main:abelian} (resp. \Cref{thm:main:freeDirProd}) demonstrates the existence of a 2QCFA that recognizes the corresponding word problem with bounded error in expected polynomial (resp. exponential) time. We next consider the class $\widehat{\Pi}_3$. While the relationship of this class to the class of groups to which \Cref{thm:main:projEmbed} applies is unclear to us, we can show that the word problem of any group in this class can be recognized by a 2QCFA with negative one-sided \textit{unbounded} error. Let $\mathsf{coN2QCFA}(T,d,\mathbb{A})$ be defined as in \Cref{def:2QCFAparams}, except we now only require that $\Pr[N \text{ rejects } w]>0, \forall w \not \in L$.

\begin{theorem}\label{thm:main:pi3}
	If $G \in \widehat{\Pi}_3$, then $W_G \in \mathsf{coN2QCFA}(n,2,\widetilde{\mathbb{C}})$.
\end{theorem}

\begin{remark}
	$\mathbb{Z} * \mathbb{Z}^2 \in \Sigma_2 \subseteq \widehat{\Pi}_3$. It is conjectured \cite{brough2014groups,holt2005groups} that $W_{\mathbb{Z} * \mathbb{Z}^2} \not \in \mathsf{poly{-}CFL} \cup \mathsf{coCFL}$.
\end{remark}

Lastly, we consider 2QCFA with no restrictions on their transition amplitudes, as well as the measure-once one-way quantum finite automaton (MO-1QFA) defined by Moore and Crutchfield \cite{moore2000quantum}. Let $\mathsf{coN1QFA}$ denote the class of languages recognizable with negative one-sided unbounded error by a MO-1QFA (with any constant number of states).

\begin{theorem}\label{thm:main:totallyGeneric}
	If $G$ is a finitely generated group that is virtually a subgroup of $(\PU(d))^k$, for some $d,k \in \mathbb{N}_{\geq 1}$, then $W_G \in \mathsf{coN2QCFA}(n,d,\mathbb{C}) \cap \mathsf{coN1QFA}$.
\end{theorem}

Let $\mathcal{D}$ denote the class of groups to which the preceding theorem applies (which includes all groups to which all earlier theorems apply). Let $\mathsf{S}$ denote the stochastic languages (the class of languages recognizable by PFA with strict cut-points). By \cite[Theorem 3.6]{brodsky2002characterizations}, $\mathsf{coN1QFA} \subseteq \mathsf{coS}$, which implies the following corollary.

\begin{corollary}
	If $G \in \mathcal{D}$, then $W_G \in \mathsf{coS}$.
\end{corollary}

\begin{remark}
	For many $G \in \mathcal{D}$, the fact that $W_G \in \mathsf{coS}$ was already known: $W_{F_k} \in \mathsf{coS}$, $\forall k$ \cite{brodsky2002characterizations}, which implies (by standard arguments from computational group theory, see for instance \cite{muller1983groups}) that $\forall G \in \widehat{\Pi}_2$, $W_G \in \mathsf{coS}$. However, for $G \in \mathcal{D} \setminus \widehat{\Pi}_2$, this result appears to be new.
\end{remark}

\subsection{Outline of the Paper}\label{sec:intro:outline}

The landmark result of Lipton and Zalcstein \cite{lipton1977word} showed that, if $G$ is a finitely generated linear group over a field of characteristic zero, then $W_G \in \mathsf{L}$. The key idea behind their logspace algorithm was to make use of a carefully chosen \textit{representation} of the group $G$ in order to recognize $W_G$ (see \Cref{sec:prelim:repTheory} for the needed notation and terminology from representation theory). Our 2QCFA algorithm will operate in a similar manner; however, the constraints of quantum mechanics will require us to make many modifications to their approach.

A (unitary) representation of a (topological) group $G$ is a continuous homomorphism $\rho:G\rightarrow \U(\mathcal{H})$, where $\mathcal{H}$ is a Hilbert space, and $\U(\mathcal{H})$ is the group of unitary operators on $\mathcal{H}$. The Gel'fand-Raikov theorem states that the elements of any locally compact group $G$ are separated by its unitary representations; i.e., $\forall g \in G$ with $g \neq 1_G$, there is some $\mathcal{H}$ and some $\rho:G \rightarrow \U(\mathcal{H})$ such that $\rho(g) \neq \rho(1_G)$. For certain groups, stronger statements can be made; in particular, one calls a group maximally almost periodic if the previous condition still holds when $\mathcal{H}$ is restricted to be finite-dimensional. 

The core idea of our approach to recognizing the word problem $W_G$ of a particular group $G$ is to construct what we have chosen to call a \textit{distinguishing family of representations} (DFR) for $G$, which is a refinement of the above notion. Informally, a DFR is a collection of a small number of unitary representations of $G$, all of which are over a Hilbert space of small dimension, such that, for any $g \in G$ other than $1_G$, there is some representation $\rho$ in the collection for which $\rho(g)$ is ``far from'' $\rho(1_G)$, relative to the ``size'' of $g$.

In Section~\ref{sec:dfr}, we formally define DFRs, and construct DFRs for many groups. Our constructions of DFRs crucially rely on certain results concerning Diophantine approximation, both in the traditional setting of approximation of real numbers by rational numbers, as well as in a certain non-commutative generalization, originally proposed by Gamburd, Jakobson, and Sarnak \cite{gamburd1999spectra}; we study Diophantine approximation in Section~\ref{sec:dfr:diophantine}. In Section~\ref{sec:2QCFAword}, we use a DFR for a group $G$ to construct a 2QCFA that recognizes $W_G$, where the parameters of the DFR directly determine the parameters of the 2QCFA. In Section~\ref{sec:discussion:wordComplex}, we compare our results to existing results regarding both the classical and quantum computational complexity of the word problem. A key feature of the 2QCFA that we construct is that they operate by storing an amount of information that grows (quite quickly) with the size of the input using only a quantum register of constant size. In Section~\ref{sec:discussion:compress}, we discuss why this is possible, and consider further implications of this extreme compression of information. 

\section{Preliminaries}\label{sec:prelim}

\subsection{Quantum Computation and the 2QCFA}\label{sec:prelim:quant}

In this section, we briefly recall the fundamentals of quantum computation and the definition of 2QCFA. For further background on quantum computation, see, for instance, \cite{nielsen2002quantum,watrous2018theory}.

A natural way of understanding quantum computation is as a generalization of probabilistic computation. One may consider a probabilistic system defined over some finite set of states $C=\{c_1,\ldots,c_k\}$, where the state of that system, at any particular point in time, is given by a probability distribution over $C$. Such a probability distribution may be described by a vector $v=(v_{c_1},\ldots,v_{c_k})$, where $v_c \in \mathbb{R}_{\geq 0}$ denotes the probability that the system is in state $c \in C$, and $\sum_c v_c=1$, i.e., $v$ is simply an element of $\mathbb{R}_{\geq 0}^k$ with $L^1$-norm $1$.

Similarly, consider some finite set of \textit{quantum basis states} $Q=\{q_1,\ldots,q_k\}$, which correspond to an orthonormal basis $\ket{q_1},\ldots,\ket{q_k}$ of $\mathbb{C}^k$ (here and throughout the paper we use the standard bra-ket notation). The state of a quantum system over $Q$, at any particular time, is given by some \textit{superposition} $\ket{\psi}=\sum_q \alpha_q \ket{q}$ of the basis states, where each $\alpha_q \in \mathbb{C}$ and $\sum_q \lvert \alpha_q \rvert^2 =1$; i.e., a superposition $\ket{\psi}$ is simply an element of $\mathbb{C}^k$ with $L^2$-norm $1$. 

Let $\U(k)$ denote the group of $k \times k$ unitary matrices. Given a quantum system currently in the superposition $\ket{\psi}$, one may apply a transformation $t \in \U(k)$ to the system, after which the system is in the superposition $t \ket{\psi}$. One may also perform a \textit{projective measurement in the computational basis}, which is specified by some partition $B=\{B_0,\ldots,B_l\}$ of $Q$. Measuring a system that is in the superposition $\ket{\psi}=\sum_q \alpha_q \ket{q}$ with respect to $B$ gives the result $B_r \in B$ with probability $p_r:=\sum_{q \in B_r} \lvert \alpha_q \rvert^2$; additionally, if the result of the measurement is $B_r$, then the state of the system \textit{collapses} to the superposition $\frac{1}{\sqrt{p_r}}\sum_{q \in B_r} \alpha_q \ket{q}$. We emphasize that measuring a quantum system changes the state of that system. 

We now define a 2QCFA, essentially following the original definition in \cite{ambainis2002two}. Informally, a 2QCFA is a two-way deterministic finite automaton that has been augmented with a finite size quantum register. Formally, a 2QCFA $M$ is given by an $8$-tuple, $M=\{Q,C,\Sigma,\delta,q_{start},c_{start},c_{acc},c_{rej}\}$, where $Q$ (resp. $C$) is the finite set of quantum (resp. classical) states, $\Sigma$ is a finite alphabet, $\delta$ is the transition function, $q_{start} \in Q$ (resp. $c_{start} \in C$) is the quantum (resp. classical) start state, and $c_{acc},c_{rej} \in C$, where $c_{acc} \neq c_{rej}$, are the accepting and rejecting states. The \textit{quantum register} of $M$ is given by the quantum system with basis states $Q$. We define the tape alphabet $\Gamma:=\Sigma \sqcup \{\#_L,\#_R\}$ where the two distinct symbols $\#_L,\#_R \not \in \Sigma$ will be used to denote, respectively, a left and right end-marker. 

Each step of the computation of the 2QCFA $M$ involves either performing a unitary transformation or a projective measurement on its quantum register, updating the classical state, and moving the tape head. This behavior is encoded in the transition function $\delta$. For each $(c,\gamma) \in (C \setminus  \{c_{acc} , c_{rej}\}) \times \Gamma$, $\delta(c,\gamma)$ specifies the behavior of $M$ when it is in the classical state $c$ and the tape head currently points to a tape alphabet symbol $\gamma$. There are two forms that $\delta(c,\gamma)$ may take, depending on whether it encodes a unitary transformation or a projective measurement. In the first case, $\delta(c,\gamma)$ is a triple $(t,c',h)$ where $t \in \U(\lvert Q \rvert)$ is a unitary transformation to be performed on the quantum register, $c' \in C$ is the new classical state, and $h \in \{-1,0,1\}$ specifies whether the tape head is to move left, stay put, or move right, respectively. In the second case, $\delta(c,\gamma)$ is a pair $(B,f)$, where $B$ is a partition of $Q$ specifying a projective measurement, and $f:B \rightarrow C \times \{-1,0,1\}$ specifies the mapping from the result of that measurement to the evolution of the classical part of the machine, where, if the result of the measurement is $B_r$, and $f(B_r)=(c',h)$, then $c' \in C$ is the new classical state and $h \in \{-1,0,1\}$ specifies the movement of the tape head.

The computation of $M$ on an input $w \in \Sigma^*$ is then defined as follows. If $w$ has length $n$, then the tape will be of size $n+2$ and contain the string $\#_L w \#_R$. Initially, the classical state is $c_{start}$, the quantum register is in the superposition $\ket{q_{start}}$, and the tape head points to the leftmost tape cell.  At each step of the computation, if the classical state is currently $c$ and the tape head is pointing to symbol $\gamma$, the machine behaves as specified by $\delta(c,\gamma)$. If, at some point in the computation, $M$ enters the state $c_{acc}$ (resp. $c_{rej}$) then it immediately halts and accepts (resp. rejects) the input $w$. As quantum measurement is a probabilistic process, the computation of $M$ is probabilistic. For any $w \in \Sigma^*$, we write $\Pr[M \text{ accepts } w]$ (resp. $\Pr[M \text{ rejects } w]$) for the probability that $M$ will accept (resp. reject) the input $w$.  

Let $\mathcal{T}=\{t \in \U(\lvert Q \rvert):\exists (c,\gamma) \in ((C \setminus \{c_{acc} , c_{rej}\}) \times \Gamma) \text{ such that } \delta(c,\gamma)=(t,\cdot,\cdot)\}$ denote the set of all unitary transformations that $M$ may perform. The \textit{transition amplitudes} of $M$ are the set of numbers $\mathbb{A}$ that appear as entries of some $t \in \mathcal{T}$.

\subsection{Group Theory and the Word Problem of a Group}\label{sec:prelim:wordProblem}

Informally, the \textit{word problem} for a group $G$ is the following question: given a finite sequence of elements $g_1,\ldots,g_n \in G$, is $g_1\cdots g_n$, their combination using the group operation, equal to the identity element of $G$? In this section, we formalize this problem. 

We begin by formally defining the word problem of a group; for more extensive background, see, for instance, \cite{loh2017geometric}. Let $F(S)$ denote the free group on the set $S$. For sets $S$ and $R$, where $R \subseteq F(S)$, let $\langle R^{F(S)}\rangle$ denote the normal closure of $R$ in $F(S)$; we say that a group $G$ has \textit{presentation} $\langle S|R \rangle$ if $G \cong F(S)/\langle R^{F(S)}\rangle$, in which case we write $G=\langle S|R \rangle$. For a set $S$, we define the set of formal inverses $S^{-1}$, such that for each $s \in S$, there is a unique corresponding $s^{-1} \in S^{-1}$, and $S \cap S^{-1} = \emptyset$. 

\begin{definition2}
	Suppose $G=\langle S|R \rangle$, where $S$ is finite. Let $\Sigma=S \sqcup S^{-1}$, let $\Sigma^*$ denote the free monoid over $\Sigma$, let $\phi:\Sigma^* \rightarrow G$ denote the natural monoid homomorphism that takes each string in $\Sigma^*$ to the element of $G$ that it represents, and let $1_G$ denote the identity element of $G$. The word problem of $G$ with respect to the presentation $\langle S|R \rangle$ is the language $W_{G=\langle S|R \rangle}=\{w \in \Sigma^*:\phi(w)=1_G\}$ consisting of all strings that represent $1_G$. 
\end{definition2} 

If $G=\langle S|R \rangle$, then $S$ (or more precisely the image of $S$ in $G$ under $\phi$) is a generating set for $G$, and if $G$ has generating set $S$, then it has (many) presentations of the form $G=\langle S|R \rangle$. We say that $G$ is \textit{finitely generated} if it has a generating set that is finite, and we say that $G$ is \textit{finitely presented} if it has a presentation $G=\langle S|R \rangle$ with both $S$ and $R$ finite. 

Note that, while the above definition of the word problem of a group $G$ does depend on the particular presentation used, the computational complexity of the word problem of $G$ does not depend on the choice of presentation (with finite generating set). To clarify this, let $\mathcal{C}$ denote a complexity class. We say that $\mathcal{C}$ is \textit{closed under inverse homomorphism} if, for all pairs of finite alphabets $\Sigma_1,\Sigma_2$, all monoid homomorphisms $\tau:\Sigma_1^* \rightarrow \Sigma_2^*$, and every language $L \in \mathcal{C}$ over the alphabet $\Sigma_2$, we have $\tau^{-1}(L)=\{v \in \Sigma_1^*:\tau(v)\in L\} \in \mathcal{C}$. For any class of languages $\mathcal{C}$ closed under inverse homomorphism, if $\langle S|R \rangle$ and $\langle S'|R' \rangle$, with $S$ and $S'$ finite, are both presentations of the same group $G$, then $W_{G=\langle S|R \rangle} \in \mathcal{C} \Leftrightarrow W_{G=\langle S'|R' \rangle} \in \mathcal{C}$ \cite{herbst1991subclass}. As each complexity class $\mathcal{C}$ considered in this paper is closed under inverse homomorphism, we will use $W_G$ to denote the word problem of a finitely generated group $G$, and we will write $W_G \in \mathcal{C}$ if $W_{G=\langle S|R \rangle} \in \mathcal{C}$ for some (equivalently, every) presentation $\langle S|R\rangle$ of $G$ with $S$ finite.

We conclude this section with a bit of additional terminology and notation from group theory needed in later parts of the paper. For a group $G$, we write $S \subseteq G$ if the set $S$ is a subset of $G$ and $H \leq G$ if the group $H$ is a subgroup of $G$. We say that a group $F$ is \textit{free} if $F \cong F(S)$ for some set $S$, and we define the \textit{rank} of $F$ to be the cardinality of $S$. The rank of a free group is well-defined as $F(S) \cong F(T)$ if and only if $S$ and $T$ have the same cardinality. As a consequence of the same observation, there is a unique (up to isomorphism) free group of rank $k$, for any $k \in \mathbb{N}$, which allows us to speak about \textit{the} free group of rank $k$, which we denote by $F_k:=F(\{1,\ldots,k\})$. We follow the convention that $F_0=F(\emptyset)=\{1\}$, the trivial group. For a group $G$ and a subgroup $H \leq G$, we use $[G:H]$ to denote the index of $H$ in $G$; if $[G:H]$ is finite, then we say that $H$ is a \textit{finite index} subgroup of $G$. We say a group is \textit{finite} if it is finite as a set, and \textit{countable} if it is at most countably infinite as a set. Notice that any finitely generated group is necessarily countable. We say a group is \textit{cyclic} if it has a generating set consisting of a single element, \textit{abelian} if the group operation is commutative, and \textit{linear} if it is isomorphic to a subgroup of $\GL(n,k)$, where $\GL(n,k)$ denotes the group of $n \times n$ invertible matrices, over some field $k$, where the group operation is given by matrix multiplication. For any property $\mathcal{P}$ (abelian, free, etc.), we say a group is \textit{virtually} $\mathcal{P}$ if it contains a finite-index subgroup that has $\mathcal{P}$. 

For a group $G=\langle S|R \rangle$, let $\Gamma(G,S)$ denote the (right) \textit{Cayley graph} of $G$ with the respect to the generating set $\phi(S)$; it is the directed, labeled graph which has vertices $G$, and a directed edge from $g$ to $g \phi(\sigma)$ that is labeled $\sigma$, for each $g \in G$ and $\sigma \in \Sigma=S \sqcup S^{-1}$. A word $w=w_1 \cdots w_n \in \Sigma^*$, with each $w_i \in \Sigma$, specifies a path $p_w$ in $\Gamma(G,S)$ which starts at the vertex $1_G$ and, on the $i^{\text{th}}$ step, follows the edge labeled $w_i$. Notice that $\phi(w)=1_G$ if and only if the path $p_w$ terminates at the vertex $1_G$. Next, notice that, if $\langle S'|R' \rangle$ is another presentation of $G$, where $S'$ is also finite, then, $\Gamma(G,S)$ and $\Gamma(G,S')$ will not generally be isomorphic graphs; however, they will ``look the same from far away.''

To formalize this notion, recall that a \textit{metric space} is a set $X$ equipped with a map $d: X \times X \rightarrow \mathbb{R}_{\geq 0}$, where $\mathbb{R}_{\geq 0}$ denotes the non-negative real numbers, such that, $\forall x_1,x_2,x_3 \in X$, the following three properties are satisfied: $d(x_1,x_2)=0 \Leftrightarrow x_1=x_2$, $d(x_1,x_2)=d(x_2,x_1)$, and $d(x_1,x_3) \leq d(x_1,x_2)+d(x_2,x_3)$. Given two metric spaces $(X,d)$ and $(X',d')$, we say that a function $f:X \rightarrow X'$ is a \textit{bilipschitz equivalence} between them if $f$ is a bijection and $\exists C \in \mathbb{R}_{>0}$ such that, $\forall x_1,x_2 \in X$, we have $\frac{1}{C} d(x_1,x_2) \leq d'(f(x_1),f(x_2)) \leq C d(x_1,x_2)$. For a group $G=\langle S|R \rangle$, the \textit{word metric on} $G$ \textit{relative to the generating set} $S$, which we denote by $d_S$, is the usual distance metric on the Cayley graph $\Gamma(G,S)$, i.e, for any $g_1,g_2 \in G$, $d_S(g_1,g_2)$ is the smallest $m \in \mathbb{N}$ for which $\exists \sigma_1,\ldots,\sigma_m \in \Sigma$ such that $g_2=g_1 \phi(\sigma_1 \cdots \sigma_m)$. Notice that $(G,d_S)$ is a metric space. It is straightforward to see that, if $S$ and $S'$ are two finite generating sets of $G$, then the identity map on $G$ is a bilipschitz equivalence between $(G,d_S)$ and $(G,d_S')$, where the constant $C$ can be straightforwardly bounded by considering $d_S$ and $d_S'$ (see, for instance, \cite[Proposition 5.2.4]{loh2017geometric}).

When $S$ is clear from context, we will often simply write $d$ in place of $d_S$. We also define $l_S(g)$, the \textit{length of} $g \in G$ \textit{relative to the generating set} $S$, by $l_S(g):=d_S(1,g)$, i.e., $l_S(g)$ is the minimum value of $m$ for which $\exists g_1,\ldots,g_m \in S \cup S^{-1}$ such that $g=\phi(g_1\cdots g_m)$. Similarly, we write $l$ in place of $l_S$, when $S$ is clear from context.

\subsection{Representation Theory Background}\label{sec:prelim:repTheory}

In this section, we state certain basic definitions and elementary results from representation theory that will be needed in the remainder of this paper. While the material in this section can be found in essentially any textbook on the (linear) representation theory of (infinite) groups, we essentially follow \cite{kowalski2014introduction}, though we deliberately avoid stating results in their full generality, to simplify the exposition as much as possible.

A \textit{representation} of a group $G$ over a field $k$ is a pair $(\rho,V_{\rho})$, where $V_{\rho}$ is a vector space over $k$, $\GL(V_{\rho})$ denotes the group of invertible $k$-linear maps on $V_{\rho}$, and $\rho:G \rightarrow \GL(V_{\rho})$ is a group homomorphism. If, furthermore, $\rho:G \rightarrow \GL(V_{\rho})$ is injective, then we say that $(\rho,V_{\rho})$ is a \textit{faithful} representation of $G$. For $v \in V_{\rho}$ and $g \in G$, we denote the image of $v$ under the map $\rho(g)$ by $\rho(g)v$. This notation is used to emphasize that a representation $(\rho,V_{\rho})$ of a group $G$ is equivalent to a linear (left) action of $G$ on $V_{\rho}$, given by $g \cdot v=\rho(g)v$, for $g \in G$ and $v \in V_{\rho}$. By standard slight abuse of notation, we will often say that $\rho$ is a representation of $G$, when $V_{\rho}$ is clear from the context. We say that $V_{\rho}$ is the \textit{representation space} of the representation $\rho$. The \textit{dimension} of a representation $\rho$ is the (vector space) dimension of its representation space $V_{\rho}$. If $\rho$ is a finite-dimensional representation, one may identify (non-canonically) $\GL(V_{\rho})$ with $\GL(n,k)$, the group of $n \times n$ invertible matrices over the field $k$, by picking a particular basis of $V$. Such an identification allows the image of $g\in G$ under the map $\rho:G \rightarrow \GL(n,k)$, to be explicitly encoded in a matrix, which will be useful for computation. 

In this paper, we concern ourselves, almost exclusively, with \textit{finite-dimensional unitary representations of finitely generated groups}, which, for such a group $G$, are representations of the form $\rho:G \rightarrow \U(n)$, for some $n \in \mathbb{N}_{\geq 1}$, where $\U(n)$ denotes the group of $n \times n$ unitary matrices, and for which the corresponding representation space $V_{\rho} = \mathbb{C}^n$. Throughout the paper, \textit{a representation} will always mean a finite-dimensional unitary representation of a finitely generated group, unless we explicitly note otherwise. 

Generally, one defines a unitary representation of a topological group $G$ as a representation $\rho:G \rightarrow \U(\mathcal{H})$, where $\mathcal{H}$ is some complex Hilbert space and $\U(\mathcal{H})$ denotes the group of all unitary continuous linear operators on $\mathcal{H}$, such that $\rho$ is strongly continuous, i.e., for every $v \in \mathcal{H}$, the mapping $G \rightarrow \mathcal{H}$ given by $g \mapsto \rho(g)v$ is continuous. However, any finitely generated group is countable, and the natural topology for any countable group is the discrete topology, for which the continuity condition is trivially satisfied. Moreover, as previously observed, finite-dimensional representations can be concretely realized as representations into matrix groups. Therefore, this is equivalent to our simpler definition.

Consider two representations $\rho_1:G \rightarrow \U(n_1)$ and $\rho_2:G \rightarrow \U(n_2)$ of a group $G$. Let $\Hom_{\mathbb{C}}(n_1,n_2)$ denote the space of $\mathbb{C}$-linear maps (i.e., homomorphisms of $\mathbb{C}$ vector spaces) $\phi:\mathbb{C}^{n_1} \rightarrow \mathbb{C}^{n_2}$. A \textit{homomorphism of representations} is a $\phi \in \Hom_{\mathbb{C}}(n_1,n_2)$ such that, $\forall g \in G, \forall v \in V_{\rho_1}=\mathbb{C}^{n_1}$, we have $\phi(\rho_1(g)v)=\rho_2(g)\phi(v)$. We use $\Hom_G(\rho_1,\rho_2)$ to denote the subspace of $\Hom_{\mathbb{C}}(n_1,n_2)$ consisting of all such $\phi$. If there is some $\phi \in \Hom_G(\rho_1,\rho_2)$ that is bijective, we say that the representations $\rho_1$ and $\rho_2$ are \textit{isomorphic}, which we denote by writing $\rho_1 \cong \rho_2$, and we call such a $\phi$ an \textit{isomorphism of representations}. For an $n_1 \times n_1$ matrix $A$ and a $n_2 \times n_2$ matrix $B$, we write $A \oplus B$ to denote the $(n_1+n_2) \times (n_1 + n_2)$ block-diagonal matrix whose two diagonal blocks are given by $A$ and $B$. The \textit{direct sum of representations} $\rho_1$ and $\rho_2$ is the representation $\rho_1 \oplus \rho_2:G \rightarrow \U(n_1 + n_2)$, where $(\rho_1 \oplus \rho_2)(g)=\rho_1(g) \oplus \rho_2(g)$, $\forall g \in G$.

For a representation $\rho:G \rightarrow \U(n)$, we say that a vector subspace $V'$ of $V_{\rho}=\mathbb{C}^n$ is \textit{stable} if $\forall g \in G, \forall v \in V'$, $\rho(g)v \in V'$. We say that the representation $\rho':G \rightarrow \U(n')$ is a \textit{subrepresentation} of $\rho$ if there is a stable subspace $V'$ of $V_{\rho}$, of dimension $n'$, such that $\rho'(g)v=\rho(g)v$, $\forall g \in G, \forall v \in V'$. We say that $\rho$ is \textit{irreducible} if it has no non-trivial subrepresentations (i.e., the only stable subspaces of $V_{\rho}$ are $0$ and $V_{\rho}$ itself). For any representation $\rho:G \rightarrow \U(n)$, there is a decomposition $\rho \cong \rho_1 \oplus \cdots \oplus \rho_m$, where the $\rho_j$ are all irreducible subrepresentations; moreover, this decomposition is unique (up to permutation of the summands, and isomorphism of representations).

For a representation $\rho:G \rightarrow \U(n)$ of a group $G$, and a subgroup $H \leq G$, we define the \textit{restricted representation} $\Res^G_H(\rho)$ to be the representation $\pi:H \rightarrow \U(n)$ of $H$, where $\pi(h)=\rho(h)$, $\forall h \in H \leq G$, i.e., this is simply the restriction of $\rho$ to $H$. Next, we define a concept dual to the notion of restriction. Let $\pi:H \rightarrow \U(m)$ be a representation of $H$ and let $G$ be a finite-index overgroup of $H$, i.e., $H \leq G$ and $r:=[G:H]$ is finite. The \textit{induced representation} $\Ind_H^G(\pi)$ is the representation $\rho:G \rightarrow \U(mr)$, which is defined as follows. Let $T=\{g_1,\ldots,g_r\} \subseteq G$ denote a complete family of left coset representatives of $H$ in $G$. Let $S_r$ denote the symmetric group on $r$ symbols. For each $g \in G$, let $\sigma_g \in S_r$ and $h_{g,j} \in H$ denote the (unique) elements such that, for each $j \in \{1,\ldots,r\}$, we have $gg_j=g_{\sigma_g(j)}h_{g,j}$. For each $g_j \in T$, let $g_j\mathbb{C}^m$ denote an isomorphic copy of the representation space $V_{\pi}=\mathbb{C}^m$. We then define $V_{\rho}$, the representation space of $\rho$, by $V_{\rho}=\bigoplus_{j=1}^r g_j \mathbb{C}^m \cong \mathbb{C}^{mr}$. To define $\rho$, we think of an element of $V_{\rho}$ as being of the form $\sum_{j=1}^r g_j v_j$, where each $v_j \in V_{\pi}=\mathbb{C}^m$, and define $\rho:G \rightarrow \U(mr)$ such that $\forall g \in G$, $\rho(g) \sum_{j=1}^r g_j v_j=\sum_{j=1}^r g_{\sigma_g(j)}\pi(h_{g,j})v_j$. Concretely, $\rho(g)$ is a block matrix, all of whose blocks are $m \times m$, and, in block-column $j$, the only non-zero block-row is $\sigma_g(j)$, and this block is given by $\pi(h_{g,j})$. 

Induction and restriction, as defined above are dual in the following sense: If one lets $\Rep_G$ (resp. $\Rep_H$) denotes, the category of representations of $G$ (resp. $H$) over the field $k$, then $\Res_H^G:\Rep_G \rightarrow \Rep_H$ and $\Ind_H^G:\Rep_H \rightarrow \Rep_G$ are functors and $\Ind_H^G$ is the left-adjoint of $\Res_H^G$. We note that induction, as we have defined it, is more commonly called co-induction, and that one traditionally defines the induced representation such that induction is the right-adjoint of restriction. However, as we only consider the case when $H$ is a finite index subgroup of $G$, the co-induced representation that we have defined and the induced representation that one normally defines are isomorphic. It will simply be more convenient, for our purposes, to use co-induction, though we will refer to it as induction.

Consider a representation $\rho:G \rightarrow \U(n)$. The \textit{character} of $\rho$ is the function $\chi_{\rho}:G \rightarrow \mathbb{C}$ given by $\chi_{\rho}(g)=\Tr(\rho(g))$, where $\Tr(\rho(g))$ denotes the trace of (the unitary matrix) $\rho(g)$. Let $I_d \in \U(d)$ denote the $d \times d$ identity matrix (i.e., the identity element of the group $\U(d)$), $Z(\U(d))=\{e^{ir}I_d|r \in \mathbb{R}\}$ denote the center of $\U(d)$, $\PU(d)=\U(d)/Z(\U(d))$ denote the $d$-dimensional projective unitary group, and $\tau:\U(d) \rightarrow \PU(d)$ denote the canonical projection. Let $\Pker(\rho)=\{g \in G|\rho(g) \in Z(\U(d))\}$ denote the quasikernel of $\rho$; notice that $\Pker(\rho)=\ker(\tau \circ \rho)$, and $\ker(\rho) \leq \Pker(\rho) \leq G$. We say that a representation $\rho$ of $G$ is \textit{projectively faithful} or simply \textit{P-faithful} if $\Pker(\rho)$ is the trivial group (i.e., if only the identity element of $G$ belongs to $\Pker(\rho)$). Notice that a P-faithful representation is necessarily a faithful representation. Furthermore, notice that, $\forall g \in G$, $\lvert \chi_{\rho}(g) \rvert \leq d$, and $\lvert \chi_{\rho}(g) \rvert = d \Leftrightarrow g \in \Pker(\rho)$. Lastly, we define a \textit{projective unitary representation of a finitely generated group} $G$ to be a group homomorphism $\pi: G \rightarrow \PU(d)$. We will use the term \textit{projective representation} to refer to such a representation.

\section{Distinguishing Families of Representations}\label{sec:dfr}

Our primary tool for constructing a 2QCFA for the word problem for a group $G$ is a \textit{distinguishing family of representations} (DFR) for the group $G$. Informally, a DFR for a group $G$ is a ``small" family of ``small" unitary representations of $G$ such that, for each $g \in G$ where $g \neq 1_G$, the family contains at least one representation which ``strongly" separates $g$ from $1_G$. The following definition formalizes this, by introducing parameters to quantify the above fuzzy notions. In this definition, and in the remainder of the paper, let $G_{\neq 1}=G \setminus \{1_G\}$, let $M(d,\mathbb{A})$ denote the set of $d \times d$ matrices with entries in some set $\mathbb{A}$, and let $\U(d,\mathbb{A})=\U(d) \cap M(d,\mathbb{A})$.

\begin{definition2}\label{def:dfr}
	Consider a group $G=\langle S|R \rangle$, with $S$ finite. For $k \in \mathbb{N}_{\geq 1}, d \in \mathbb{N}_{\geq 2}$, $\tau:\mathbb{R}_{>0} \rightarrow \mathbb{R}_{>0}$ a monotone non-increasing function, and $\mathbb{A} \subseteq \mathbb{C}$, we define a $[k,d,\tau,\mathbb{A}]$-\textit{distinguishing family of representations} (DFR) for $G$ to be a set $\mathcal{F}=\{\rho_1,\ldots,\rho_k\}$ where the following conditions hold.
	\begin{enumerate}[(a)]
		\item \label{def:dfr:unitary} $\forall j \in \{1,\ldots,k\}, \ \rho_j:G \rightarrow \U(d)$ is a representation of $G$.
		\item \label{def:dfr:distinguish} $\forall g \in G_{\neq 1}$, $\exists j \in \{1,\ldots,k\}$ such that $\lvert \chi_{\rho_j}(g)\rvert \leq d-\tau(l(g))$.
		\item \label{def:dfr:amplitudes} $\forall \sigma \in S \cup S^{-1}, \forall j \in \{1,\ldots,k\}$, $\exists Y_1,\ldots,Y_t \in \U(d,\mathbb{A})$, such that $\rho_j(\sigma)=\prod_i Y_i$.
	\end{enumerate}
\end{definition2} 

Suppose $\mathcal{F}=\{\rho_1,\ldots,\rho_k\}$ is a $[k,d,\tau,\mathbb{A}]$-DFR for $G=\langle S|R \rangle$. We write $I_d =1_{\U(d)} \in \U(d)$ for the $d \times d$ identity matrix, $\ker(\rho_j)=\{g \in G:\rho_j(g)=I_d\}$ for the kernel of $\rho_j$, $Z(\U(d))=\{e^{ir} I_d:r \in \mathbb{R}\}$ for the center of $\U(d)$, and $\Pker(\rho_j)=\{g \in G: \rho_j(g)=Z(\U(d))\}$ for the quasikernel of $\rho_j$. Clearly, $1_G \in \Pker(\rho_j), \forall j$, but, as $\rho_j$ is not assumed to be P-faithful or even faithful, there may be $g \in G_{\neq 1}$ for which, for certain $j$, we have $g \in \Pker(\rho_j)$. However, due to the fact that $g \in \Pker(\rho_j)$ exactly when $\lvert \chi_{\rho_j}(g) \rvert = d$, the second defining property of a DFR guarantees not only that $\bigcap_j \Pker(\rho_j) = \{1_G\}$, but, much more strongly, that all $g \in G_{\neq 1}$ are ``far from'' being in $\bigcap_j \Pker(\rho_j)$. That is to say, $\forall g \in G_{\neq 1},\exists j$ such that $\lvert \chi_{\rho_j}(g) \rvert$ is at distance at least $\tau(l(g))$ from having value $d$. The following proposition is then immediate, but we explicitly state it as it is the central notion in our quantum approach to the word problem.

\begin{proposition}\label{thm:dfr:intersectPKer}
	Suppose $G=\langle S| R \rangle$ has a $[k,d,\tau,\mathbb{A}]$-DFR $\mathcal{F}=\{\rho_1,\ldots,\rho_k\}$. Then, $\forall g \in G$, 
	$g=1_G \Leftrightarrow \forall j, \ \lvert \chi_{\rho_j}(g)\rvert =d$ and $g \in G_{\neq 1} \Leftrightarrow \exists j \text{ such that } \lvert \chi_{\rho_j}(g) \rvert \leq d-\tau(l(g))$.
	
\end{proposition}

Note that, in the preceding proposition, $\rho_1 \oplus \cdots \oplus \rho_k:G \rightarrow \U(kd)$ is simply a faithful representation of $G$, decomposed into subrepresentations in a convenient way. Next, we establish some terminology that will better allow us to describe particular types of DFR.

\begin{definition2}\label{def:dfrTypes} Suppose $\mathcal{F}=\{\rho_1,\ldots,\rho_k\}$ is a $[k,d,\tau,\mathbb{A}]$-DFR for a group $G$. 
	\begin{enumerate}[(a)]
		\item \label{def:dfrTypes:algebraic} If $\mathbb{A}=\overline{\mathbb{Q}}$ (equivalently, if $\rho_j(G) \subseteq \U(d,\overline{\mathbb{Q}}), \forall j$), we say $\mathcal{F}$ is an \textit{algebraic} DFR.
		\item \label{def:dfrTypes:diagonal} If $\rho_j(g)$ is a diagonal matrix $\forall j, \forall g$, then we say $\mathcal{F}$ is a \textit{diagonal} DFR. 
		\item \label{def:dfrTypes:virtual} If $H$ is a finite-index overgroup of $G$, we say that $H$ \textit{virtually} has a $[k,d,\tau,\mathbb{A}]$-DFR.
	\end{enumerate}
\end{definition2}

When $\mathcal{F}$ is an algebraic DFR, we will often only write $[k,d,\tau]$ to denote its parameters. Note that only abelian groups have diagonal DFRs, and any DFR of an abelian group can be converted to a diagonal DFR; we define diagonal DFRs for convenience. 

Using a $[k,d,\tau,\mathbb{A}]$-DFR for a group $G$, it will be possible to construct a 2QCFA that recognizes the word problem $W_H$ of any finite-index overgroup $H$ of $G$, where the parameters of the DFR will strongly impact the parameters of the resulting 2QCFA. In particular, in \Cref{sec:2QCFAword}, we produce a 2QCFA with $d$ quantum states and transition amplitudes in $\mathbb{A}$ that recognizes $W_H$, with expected running time approximately $O(\tau(n)^{-1})$. The goal is then to show that a wide collection of groups virtually have DFRs with good parameters.

\subsection{Diophantine Approximation}\label{sec:dfr:diophantine}

Our constructions of DFRs rely crucially on certain results concerning Diophantine approximation. Most fundamentally, the Diophantine approximation question asks how well a particular real number $\alpha$ can be approximated by rational numbers. Of course, as $\mathbb{Q}$ is dense in $\mathbb{R}$, one can choose $\frac{p}{q} \in \mathbb{Q}$ so as to make the quantity $\lvert \alpha-\frac{p}{q} \rvert$ arbitrarily small; for this reason, one considers $\frac{p}{q}$ to be a ``good'' approximation to $\alpha$ only when $\lvert \alpha-\frac{p}{q} \rvert$ is small compared to a suitable function of $q$. One then considers $\alpha$ to be poorly approximated by rationals if, for some ``small'' constant $d \in \mathbb{R}_{\geq 2}$, $\exists C \in \mathbb{R}_{>0}$ such that, $\forall (p,q) \in \mathbb{Z} \times \mathbb{Z}_{\neq 0}$, we have $\lvert \alpha - \frac{p}{q}\rvert \geq C\lvert q\rvert^{-d}$, where the smallness of $d$ determines just how poorly approximable $\alpha$ is. For $\alpha \in \mathbb{R}$, let $\lVert \alpha \rVert=\min_{m \in \mathbb{Z}} \lvert \alpha - m \rvert$ denote the distance between $\alpha$ and its nearest integer. Notice that $\left\lvert \alpha - \frac{p}{q}\right\rvert \geq C\lvert q\rvert^{-d}, \ \forall (p,q) \in \mathbb{Z} \times \mathbb{Z}_{\neq 0} \Leftrightarrow  \lVert q \alpha \rVert \geq C\lvert q\rvert^{-(d-1)}, \ \forall q \in \mathbb{Z}_{\neq 0}$. Of particular relevance to us is the following result, due to Schmidt \cite{schmidt1970simultaneous}, that real irrational algebraic numbers are poorly approximated by rationals. 

\begin{proposition}\cite{schmidt1970simultaneous}\label{thm:schmidt}
	$\forall \alpha_1,\ldots,\alpha_k \in (\mathbb{R} \cap \overline{\mathbb{Q}})$ where $1,\alpha_1,\ldots,\alpha_k$ are linearly independent over $\mathbb{Q}$, $\forall \epsilon\in\mathbb{R}_{>0}$, $\exists C \in \mathbb{R}_{>0}$ such that $\forall q \in \mathbb{Z}_{\neq 0}$, $\exists j$ such that $\lVert q \alpha_j \rVert \geq C\lvert q\rvert^{-(\frac{1}{k}+\epsilon)}$.
\end{proposition} 

We also require the following result concerning the Diophantine properties of linear forms in logarithms of algebraic numbers, due to Baker \cite{baker1990transcendental}.

\begin{proposition}\cite{baker1990transcendental}\label{thm:baker}
	Let $L=\{\beta \in \mathbb{C}_{\neq 0}:e^{\beta} \in \overline{\mathbb{Q}}\}$. $\forall \beta_1,\ldots,\beta_k \in L$ that are linearly independent over $\mathbb{Q}$, $\exists C \in \mathbb{R}_{>0}$ such that, $\forall (q_1,\ldots,q_k) \in \mathbb{Z}^k$ with $q_{max}:=\max_j \lvert q_j\rvert>0$, $\lvert q_1 \beta_1 + \cdots + q_k \beta_k \rvert \geq (e q_{max})^{-C}$.
\end{proposition}

Gamburd, Jakobson, and Sarnak \cite[Proposition 4.3]{gamburd1999spectra} established a particular result concerning the Diophantine properties of $\SU(2,\overline{\mathbb{Q}})$, the group of $2 \times 2$ unitary matrices of determinant $1$ whose entries are algebraic numbers. The following lemma generalizes their result to $\U(d,\overline{\mathbb{Q}})$. For a group $G$, and a set of elements $S_H \subseteq G$, let $H=\langle S_H \rangle$ denote the subgroup of $G$ generated by $S_H$; for any $h \in H$, let $l(h)$ denote the length of $H$ with respect to $S_H$. Recall that the center of $\U(d,\overline{\mathbb{Q}})$ is given by $Z(\U(d,\overline{\mathbb{Q}}))=\{e^{ir}I_d:r \in \mathbb{R},e^{ir} \in \overline{\mathbb{Q}}\}$.

\begin{lemma}\label{thm:nonComAlgDiophantine}
	Consider any $S_H=\{h_1,\ldots,h_k\} \subseteq \U(d,\overline{\mathbb{Q}})$, and let $H=\langle S_H \rangle$. Then $\exists C \in \mathbb{R}_{\geq 1}$, such that $\forall h \in (H\setminus Z(\U(d,\overline{\mathbb{Q}})))$, we have $\lvert \Tr(h) \rvert \leq d- C^{-l(h)}$.
\end{lemma}
\begin{proof}
	Notice that $Z(\U(1,\overline{\mathbb{Q}}))=\U(1,\overline{\mathbb{Q}})$, and so the conclusion is vacuously true when $d=1$; we assume for the remainder of the proof that $d \geq  2$.
	
	We begin by following, essentially, the proof of \cite[Proposition 4.3]{gamburd1999spectra}. As $S_H$ is a finite subset of $M_d(\overline{\mathbb{Q}})$, there is some finite degree extension $K$ of $\mathbb{Q}$ such that $S_H \subseteq M_d(K)$. Let $\mathcal{O}_K$ denote the ring of integers of $K$ and set $N \in \mathbb{Z}_{> 0}$ sufficiently large such that $Nh_i \in M_d(\mathcal{O}_K)$, $\forall i$. Let $s$ denote the degree of $K$ over $\mathbb{Q}$, and let $\sigma_1,\ldots,\sigma_s$ denote the $s$ distinct embeddings of $K$ in $\mathbb{C}$, where $\sigma_1$ is the identity map. Each $\sigma_j:K \rightarrow \mathbb{C}$ induces a map $M_d(K) \rightarrow M_d(\mathbb{C})$ in the obvious way, which we also denote by $\sigma_j$. For a matrix $M$, let $\lVert M \rVert$ denote the Hilbert-Schmidt norm (i.e., $\lVert M \rVert^2=\sum_{i,j} \lvert M_{ij} \rvert^2$). Let $B=\max_{i,j} \lVert \sigma_j(h_i) \rVert$, and notice that $B \geq \sqrt{d}$ as $h_j \in \U(d)$ implies $\lVert \sigma_1(h_j) \rVert=\lVert h_j \rVert=\sqrt{d}$. 
	
	Fix $h \not \in Z(\U(d,\overline{\mathbb{Q}}))$. In particular, $h \neq I_d=1_H$, and so $l(h) \geq 1$. As $\lVert \cdot \rVert$ is submultiplicative, we then have $\lVert \sigma_j(h) \rVert \leq B^{l(h)}$, $\forall j$. For $r,c \in \{1,\ldots,d\}$, and $W$ a $d \times d$ matrix, we write $W[r,c]$ to denote the entry of $W$ in row $r$ and column $c$. 
	
	There are two cases. First, suppose there is some $r$ such that $h[r,r] \neq h[1,1]$. Fix such an $r$. Let $y$ denote the $d \times d$ matrix given by $y=h-h[1,1]I_d$ and notice that $y[r,r]=h[r,r]-h[1,1] \neq 0$. For every $j$, we have $$\lvert \sigma_j(y[r,r]) \rvert = \lvert \sigma_j(h[r,r])-\sigma_j(h[1,1]) \rvert \leq \lvert \sigma_j(h[r,r])\rvert + \lvert \sigma_j(h[1,1]) \rvert \leq 2\lVert \sigma_j(h) \rVert \leq 2B^{l(h)}.$$
	
	By construction, $N^{l(h)} h \in M_d(\mathcal{O}_K)$, $\forall h \in H=\langle S_H \rangle$, which immediately implies $N^{l(h)} y=N^{l(h)}(h-h[1,1]I_d) \in M_d(\mathcal{O}_K)$. Therefore, $N^{l(h)} y[r,r]$ is some non-zero element of $\mathcal{O}_K$, which implies $\prod_j \sigma_j(N^{l(h)}y[r,r]) \in \mathbb{Z}_{\neq 0}$. By the above, $\lvert \sigma_j(N^{l(h)}y[r,r]) \rvert \leq 2(BN)^{l(h)}\leq (2BN)^{l(h)}$, $\forall j$. Therefore, $$\lvert y[r,r] \rvert=\lvert \sigma_1(y[r,r]) \rvert =N^{-l(h)} \lvert \sigma_1(N^{l(h)}y[r,r]) \rvert \geq N^{-l(h)}\frac{1}{\prod_{j>1} \lvert \sigma_j( N^{l(h)}y[r,r]) \rvert} \geq ((2B)^{d-1}N^d)^{-l(h)}.$$ Notice that $$\lvert h[r,r] + h[1,1] \rvert^2+\lvert h[r,r]- h[1,1] \rvert^2=2\lvert h[r,r] \rvert^2 + 2\lvert h[1,1] \rvert^2 \leq 4.$$ Therefore, $$\lvert h[r,r] + h[1,1] \rvert 
	\leq \sqrt{4-\lvert  h[r,r] - h[1,1]\rvert^2} \leq 2-\frac{1}{4} \lvert  h[r,r] - h[1,1]\rvert^2=2-\frac{1}{4} \lvert  y[r,r] \rvert^2 \leq 2-C^{-l(h)},$$ where $C=((2BN)^{2d}) \geq 1$ (notice $l(h) \geq 1$, $B \geq \sqrt{d} \geq 1$, and $N \geq 1$). Therefore, $$\lvert \Tr(h) \rvert=\left\lvert \sum_i h[i,i] \right\rvert \leq \lvert h[r,r]+h[1,1] \rvert + \left\lvert \sum_{i \not \in \{1,r\}}  h[i,i] \right\rvert \leq 2-C^{-l(h)} + (d-2)=d-C^{-l(h)}.$$
	
	Next, suppose instead $h[r,r]=h[1,1]$, $\forall r$. As $h \not \in Z(\U(d,\overline{\mathbb{Q}}))$, there must then be some $r,c \in \{1,\ldots,d\}$, $r \neq c$, such that $h[r,c] \neq 0$ (if there were no such $r,c$, then $h=h[1,1]I_d \in Z(\U(d,\overline{\mathbb{Q}}))$). Fix such a pair $r,c$. For every $j$, we have $$\lvert \sigma_j(h[r,c]) \rvert \leq \lVert \sigma_j(h) \rVert \leq B^{l(h)}.$$ Furthermore, $N^{l(h)}h[r,c]$ is some non-zero element of $\mathcal{O}_K$, and so $$\lvert h[r,c] \rvert =N^{-l(h)} \lvert \sigma_1(N^{l(h)} h[r,c]) \rvert \geq N^{-l(h)}\frac{1}{\prod_{j>1} \lvert \sigma_j( N^{l(h)}h[r,c]) \rvert} \geq (B^{d-1}N^d)^{-l(h)}.$$ As $\lvert h[r,r] \rvert^2+\lvert h[r,c] \rvert^2 \leq 1$, we have $$\lvert h[r,r] \rvert \leq \sqrt{1-\lvert h[r,c] \vert^2} \leq 1-\frac{1}{2} \lvert h[r,c] \rvert^2\leq 1-C^{-l(h)}.$$ Therefore, \[\lvert \Tr(h) \rvert=\left\lvert \sum_i h[i,i] \right\rvert \leq \lvert h[r,r]\rvert + \left\lvert \sum_{i \neq r}  h[i,i] \right\rvert \leq 1-C^{-l(h)} + (d-1)=d-C^{-l(h)}.\qedhere \]
\end{proof}

By expressing the above condition in the language of representation theory, we have the following.

\begin{corollary}\label{thm:characterAlgDiophantine}
	Consider a group $G=\langle S|R \rangle$, with $S$ finite, and a representation $\rho:G \rightarrow \U(d,\overline{\mathbb{Q}})$. Then $\exists C \in \mathbb{R}_{\geq 1}$ such that $\forall g \in (G \setminus \Pker(\rho))$, we have $\lvert \chi_{\rho}(g) \rvert \leq d-C^{-l(g)}$. 
\end{corollary}

\subsection{Constructions of DFRs}\label{sec:dfr:constructions}

We now show that a wide collection of groups virtually have DFRs with good parameters. We accomplish this by first constructing DFRs for only a small family of special groups. We then present several constructions in which a DFR for a group, or more generally a family of DFRs for a family of groups, is used to produce a DFR for a related group. 

We begin with a straightforward lemma expressing a useful character bound. In this lemma, and throughout this section, we continue to write group operations multiplicatively, and so, for $g \in G$ and $h \in \mathbb{Z}$, if $h>0$ (resp. $h<0$) then $g^h$ denotes the element of $G$ obtained by combining $h$ copies of $g$ (resp. $g^{-1}$) with the group operation, and if $h=0$ then $g^h=1_G$. Let $S_1=\{e^{ir}:r \in \mathbb{R}\} \subseteq \mathbb{C}^{*}$ denote the circle group and let $\T(d) \subseteq \U(d)$ denote the group of all $d \times d$ diagonal matrices where each diagonal entry lies in $S_1$. For $\mathbb{A} \subseteq \mathbb{C}$, let $S_1(\mathbb{A})=S_1 \cap \mathbb{A}$ and $\T(d,\mathbb{A})=\T(d) \cap M(d,\mathbb{A})$. Let $\mathbf{1}_d:G \rightarrow \U(d)$ denote the trivial representation of dimension $d$ (i.e., $\mathbf{1}_d(g)=I_d=1_{\U(d)}$, $\forall g \in G$). For a cyclic group $G=\langle a|R_G \rangle$ and for some $r \in \mathbb{R}$, define the representation $\widehat{\gamma}_r:G \rightarrow S_1 \cong \U(1)$ such that $a \mapsto e^{2\pi i r}$; furthermore, define the representation $\gamma_r:G \rightarrow \T(2)$ by $\gamma_r=\widehat{\gamma}_r \oplus \mathbf{1}_1$.

\begin{lemma}\label{thm:dfr:diagCharBound}
	Consider the cyclic group $G=\langle a|R_G \rangle$. Fix $r \in \mathbb{R}$ and define $\gamma_r:G \rightarrow \T(2)$ as above. Suppose that $h \in \mathbb{Z}$ and $\epsilon \in \mathbb{R}_{>0}$ satisfy $\lVert hr \rVert \geq \epsilon$. Then $\chi_{\gamma_r}(a^h) \leq 2-\frac{19\pi^2}{24}\epsilon^2$.
\end{lemma}
\begin{proof}
	We have $\chi_{\gamma_r}(a^h)=e^{2\pi ihr}+1=2e^{\pi ihr}\cos(\pi hr)$. Clearly, $\epsilon \leq \frac{1}{2}$. Therefore, \[\lvert \chi_{\gamma_r}(a^h) \rvert=2\lvert \cos(\pi hr) \rvert \leq 2 \cos(\pi \epsilon)\leq 2\left(1-\frac{(\pi \epsilon)^2}{2}  +\frac{(\pi\epsilon)^4}{24}\right)\leq 2-\frac{19\pi^2}{24}\epsilon^2. \qedhere \]
\end{proof}

We first construct DFRs for a very narrow class of special groups: (i) $\mathbb{Z}_m=\langle a| a^m \rangle$, the integers modulo $m$, where the group operation is addition, (ii) $\mathbb{Z}=\langle a| \rangle$, the integers, where the group operations is addition, and (iii) $F_2=\langle a,b |\rangle$ the (non-abelian) free group of rank $2$. 

\begin{lemma}\label{thm:dfr:init:Z/mZ} 
	$\mathbb{Z}_m=\langle a| a^m \rangle$ has a diagonal algebraic $\left[1,2,\frac{19 \pi^2}{24 m^2}\right]$-DFR, $\forall m \in \mathbb{N}_{\geq 2}$.	
\end{lemma}
\begin{proof}
	Fix $m \in \mathbb{N}_{\geq 2}$ and let $r=\frac{1}{m}$. Define $\gamma_r:\mathbb{Z}_m \rightarrow \T(2)$ as above, and notice that $\gamma_r(\mathbb{Z}_m) \subseteq \T(2,\overline{\mathbb{Q}})$. Consider any $q \in \mathbb{Z}_m$, where $q \not \equiv 0 \mod m$. Then $q$ can be expressed as $q=a^h$, for $h\in \mathbb{Z}$, $h \not \equiv 0 \mod m$. As $\lVert hr \rVert \geq \frac{1}{m}$, \Cref{thm:dfr:diagCharBound} implies $\lvert \chi_{\gamma_r}(q) \rvert \leq 2-\frac{19 \pi^2}{24 m^2}$. Therefore, $\{\gamma_r\}$ is a diagonal algebraic DFR for $\mathbb{Z}_m$, with the desired parameters.   
\end{proof}

\begin{lemma}\label{thm:dfr:init:Z:nonAlg} 
	$\forall \delta \in \mathbb{R}_{>0}, \exists C \in \mathbb{R}_{>0}$, $\mathbb{Z}=\langle a| \rangle$ has a diagonal $[1+\lfloor \frac{2}{\delta} \rfloor,2,Cn^{-\delta},\widetilde{\mathbb{C}}]$-DFR.	
\end{lemma}
\begin{proof}
	Let $k=1+\lfloor \frac{2}{\delta} \rfloor$ and $\eta=\frac{\delta}{2}-\frac{1}{k}>0$. Fix $\alpha_1,\ldots,\alpha_k \in (\overline{\mathbb{Q}} \cap \mathbb{R})$ such that $1,\alpha_1,\ldots,\alpha_k$ are linearly independent over $\mathbb{Q}$. For each $j \in \{1,\ldots,k\}$ define the representation $\gamma_{\alpha_j}:\mathbb{Z} \rightarrow \T(2)$ as above, and notice that $\gamma_{\alpha_j}(\mathbb{Z}) \subseteq \T(2,\widetilde{\mathbb{C}})$. By \Cref{thm:schmidt}, $\exists D \in \mathbb{R}_{>0}$, such that $\forall q \in \mathbb{Z}_{\neq 0}$ (i.e., $\forall q \in \mathbb{Z}$ where $q \neq 0=1_{\mathbb{Z}}$), $\exists j$ such that $\lVert q \alpha_j \rVert \geq D\lvert q\rvert^{-(\frac{1}{k}+\eta)}=D\lvert q\rvert^{-\frac{\delta}{2}}$. Therefore, for any $q \in \mathbb{Z}_{\neq 0}$, if we take $j$ as above, then by \Cref{thm:dfr:diagCharBound} (with $r=\alpha_j$, $\epsilon=D\lvert q\rvert^{-\frac{\delta}{2}}$, and $h=q$) we have $\lvert \chi_{\gamma_{\alpha_j}}(q) \rvert \leq 2-\frac{19 \pi^2}{24}D^2 \lvert q \rvert^{-\delta}$. Therefore, $\{\gamma_{\alpha_1},\ldots,\gamma_{\alpha_k}\}$ is a diagonal $[1+\lfloor \frac{2}{\delta} \rfloor,2,\frac{19 \pi^2}{24} D^2 n^{-\delta},\widetilde{\mathbb{C}}]$-DFR for $\mathbb{Z}$.
\end{proof}

\begin{lemma}\label{thm:dfr:init:Z:alg} 
	$\exists C_1,C_2 \in \mathbb{R}_{>0}$ such that $\mathbb{Z}=\langle a| \rangle$ has a diagonal algebraic $[1,2,C_2 n^{-C_1}]$-DFR.	
\end{lemma}
\begin{proof}
	As in \Cref{thm:baker}, let $L=\{\beta \in \mathbb{C}_{\neq 0}:e^{\beta} \in \overline{\mathbb{Q}}\}$ and notice that $\pi i \in L$. Let $R=\{r \in ((\mathbb{R} \setminus \mathbb{Q})\cap (0,1)):2 \pi i r \in L\}$ (e.g., $\hat{r}=\frac{1}{2 \pi}\cos^{-1}\left(\frac{3}{5}\right)$ is irrational and has $e^{2 \pi i \hat{r}}=\frac{3+4i}{5}$, and so $\hat{r} \in R$). Fix $r \in R$. By definition, $2 \pi i r \in L$, which immediately implies $\pi i r \in L$. Also by definition, $r \not \in \mathbb{Q}$, which implies $\pi i r$ and $\pi i$ are linearly independent over $\mathbb{Q}$. Therefore, by \Cref{thm:baker}, $\exists D \in \mathbb{R}_{>0}$ such that $\forall (q,m) \in \mathbb{Z}^2$ where $q_{max}:=\max(\lvert q\rvert,\lvert m \rvert)>0$, we have $\lvert q\pi i r -m \pi i \rvert \geq (e q_{max})^{-D}$.
	
	For fixed $q \in \mathbb{Z}_{\neq 0}$ and varying $m \in \mathbb{Z}$, $\lvert q \pi i r -m \pi i \rvert$ attains its minimum when $m=\text{round}(qr)$, the closest integer to $qr$. Notice that $\lvert \text{round}(qr) \rvert \leq \lvert q \rvert$, as $r \in (0,1)$ and $q \in \mathbb{Z}$. Therefore, for any  $q \in \mathbb{Z}_{\neq 0}$, we have \[\lVert qr \rVert=\min_{m \in \mathbb{Z}} \lvert qr-m \rvert=\frac{1}{\pi} \min_{m \in \mathbb{Z}}\lvert q \pi i r-m \pi i \rvert=\frac{1}{\pi} \lvert q \pi i r-\text{round}(qr) \pi i \rvert \geq \frac{1}{\pi}\lvert eq \rvert^{-D}.\]
	
	Define $\gamma_r:\mathbb{Z} \rightarrow \T(2)$ as above. By \Cref{thm:dfr:diagCharBound}, $\lvert \chi_{\gamma_r}(q) \rvert \leq 2-\frac{19}{24} \lvert eq \rvert^{-2D}$. Clearly, $\gamma_r(\mathbb{Z}) \subseteq \T(2,\overline{\mathbb{Q}})$. Therefore, $\{\gamma_r\}$ is a diagonal algebraic $[1,2,\frac{19}{24} e^{-2D} n^{-2D}]$-DFR for $\mathbb{Z}$. 	
\end{proof}

\begin{remark}
	We note that the above constructions of DFRs for $\mathbb{Z}$ are quite similar to the technique used by Ambainis and Watrous \cite{ambainis2002two} to produce a 2QCFA that recognizes $L_{eq}$ (cf. \cite{brodsky2002characterizations,rabin1963probabilistic}). In particular, their approach relied on the fact that the number $\sqrt{2} \in \overline{\mathbb{Q}}$ is poorly approximated by rationals; our constructions make use of more general Diophantine approximation results. This allows us to produce 2QCFA with improved parameters.
\end{remark}

\begin{lemma}\label{thm:dfr:init:F2} 
	$\exists C \in \mathbb{R}_{\geq 1}$, such that $F_2=\langle a,b |\rangle$ has an algebraic $[1,2,C^{-n}]$-DFR.	
\end{lemma}
\begin{proof}
	First, define the representation $\pi:F_2 \rightarrow SO(3,\mathbb{Q})$ by 
	\[a \mapsto \frac{1}{5}\begin{pmatrix} 3 & -4 & 0\\4 & 3 & 0\\0 & 0 & 5\end{pmatrix}\text{ and } 
	b \mapsto \frac{1}{5}\begin{pmatrix} 5 & 0 & 0\\0 & 3 & -4\\0 & 4 & 3\end{pmatrix}.\]
	
	This is the ``standard'' faithful representation of $F_2$ into $\SO(3)$ used in many treatments of the Banach-Tarski paradox. Recall that $\SU(2)$ is the double cover of $\SO(3)$, i.e., $\SU(2)/Z(\SU(2)) \cong \SO(3)$. Then $\pi$ induces a homomorphism $\widehat{\pi}:F_2 \rightarrow \SU(2)/Z(\SU(2))$ in the obvious way, which, by the universal property of the free group, can be lifted to the representation $\rho:F_2 \rightarrow \SU(2,\overline{\mathbb{Q}})$ given by
	\[a \mapsto \frac{1}{\sqrt{5}}\begin{pmatrix} 2+i & 0\\0 & 2-i\end{pmatrix}\text{ and } 
	b \mapsto \frac{1}{\sqrt{5}}\begin{pmatrix} 2 & i\\i & 2\end{pmatrix}.\]
	
	As $\pi$ is faithful, we conclude that $\rho(g) \not\in Z(\SU(2))$, $\forall g \in (F_2 \setminus 1_{F_2})$. Therefore, by \Cref{thm:characterAlgDiophantine}, $\{\rho\}$ is an algebraic $[1,2,C^{-n}]$-DFR for $F_2$.
\end{proof}

\begin{remark}
	Note that the proof of the preceding lemma uses, fundamentally, the same construction used by Ambainis and Watrous \cite{ambainis2002two} to produce a 2QCFA for $L_{pal}$ (which is closely related to $F_2$). The algebraic structure of $F_2$ allows a substantially simpler argument.
\end{remark}

We now present several constructions of new DFRs from existing DFRs. We emphasize that all results in the following lemmas are constructive in the sense that, given the supposed DFR or collection of DFRs, each corresponding proof provides an explicit construction of the new DFR. We begin by considering conversions of a DFR of a group $G$ to a DFR with different parameters of the same group $G$. For $C \in \mathbb{R}_{>0}$, let $\eta_C:\mathbb{R}_{>0} \rightarrow \mathbb{R}_{>0}$ be given by $\eta_C(n)=Cn$.

\begin{lemma}\label{thm:dfr:conv}
	Suppose $\mathcal{F}$ is a $[k,d,\tau,\mathbb{A}]$-DFR for a group $G=\langle S|R \rangle$, with $S$ finite. The following statements hold.
	\begin{enumerate}[(i)]
		\item \label{thm:dfr:conv:combine}$G$ has a $[1,kd,\tau,\mathbb{A}]$-DFR.
		\item \label{thm:dfr:conv:increaseDimension} If $d' \in \mathbb{N}$ and $d'>d$, then $G$ has a $[k,d',\tau,\mathbb{A}]$-DFR. 
		\item \label{thm:dfr:conv:switchPresentation} Suppose $G$ also has presentation $\langle S'|R' \rangle$, with $S'$ finite. Then $\exists C \in \mathbb{R}_{>0}$ such that $\mathcal{F}$ is also a $[k,d,\tau \circ \eta_C,\mathbb{A}]$-DFR for $G=\langle S'|R' \rangle$.		
	\end{enumerate}
	Moreover, if $\mathcal{F}$ is a diagonal DFR, then each newly constructed DFR is also diagonal.
\end{lemma}
\begin{proof}
	\begin{enumerate}[(i)]
		\item Consider the representation $\rho:G \rightarrow \U(kd)$ of $G$ given by $\rho=\rho_1 \oplus \cdots \oplus \rho_k$. As $\mathcal{F}=\{\rho_1,\ldots,\rho_k\}$ is a DFR for $G$, it satisfies the property \Cref{def:dfr}(\ref{def:dfr:distinguish}); for each $g \in G_{\neq 1}$, set $j_g$ to be the corresponding value of $j \in \{1,\ldots,k\}$ provided by the property. Therefore, for each $g \in G_{\neq 1}$, we have, $$\lvert \chi_{\rho}(g) \rvert=\left\lvert \sum_j \chi_{\rho_j}(g) \right\rvert \leq \lvert \chi_{\rho_{j_g}}(g)\rvert +\left\lvert \sum_{j \neq j_g} \chi_{\rho_i}(g) \right\rvert \leq d-\tau(l(g)) +(k-1)d\leq kd-\tau(l(g)).$$ 
		
		\item For each $j$, define the representation $\widehat{\rho}_j=\rho_j \oplus \mathbf{1}_{d'-d}$. Then $\{\widehat{\rho}_1,\ldots,\widehat{\rho}_k\}$ is a $[k,d',\tau,\mathbb{A}]$-DFR, by an argument analogous to the above proof of (i).
		
		\item Let $\Gamma(G,\Sigma)$ (resp. $\Gamma(G,\Sigma')$) denote the Cayley graph of $G$ with (symmetric) generating sets $\Sigma=S \cup S^{-1}$ (resp. $\Sigma'=S' \cup S'^{-1}$). Let $d_S$ and $d_{S'}$ denote the corresponding word metrics. Then $id_G:G \rightarrow G$, the identity map on $G$, is a bilipschitz equivalence between $(G,d_S')$ and $(G,d_S)$ (see, for instance, \cite[Proposition 5.2.4]{loh2017geometric}), and so $\exists C \in \mathbb{R}_{>0}$ such that, $\forall g_1,g_2 \in G$, $\frac{1}{C} d_{S'}(g_1,g_2) \leq d_S(g_1,g_2) \leq C d_{S'}(g_1,g_2)$. We then write $l_S(g)=d_S(g,1_G)$ and $l_{S'}(g)=d_{S'}(g,1_G)$ for the length of $g \in G$ with respect to each of the generating sets $S$ and $S'$. By the above, $l_S(g) \leq C l_{S'}(g)$. As $\mathcal{F}$ is a $[k,d,\tau,\mathbb{A}]$-DFR for $G$, we have that $\forall g \in G_{\neq 1}, \exists j_g \in \{1,\ldots,k\}$ such that $\lvert \chi_{\rho_{j_g}}(g) \rvert \leq d-\tau(l_S(g))$. As $l_S(g)\leq C l_{S'}(g)$, and $\tau$ is monotone non-increasing, we then have $\tau(l_S(g)) \geq \tau(C l_{S'}(g))$, which immediately implies $\lvert \chi_{\rho_{j_g}}(g) \rvert \leq d-\tau(C l_{S'}(g))$, as desired. \qedhere		
	\end{enumerate}
\end{proof}

Next, we show that a DFR of $G$ and a DFR of $H$ can be used to produce a DFR of $G \times H$, the direct product of $G$ and $H$. In the following, for a group $Q$, let $[q_1,q_2]=q_1^{-1}q_2^{-1}q_1 q_2$ denote the commutator of elements $q_1,q_2 \in Q$. For functions $\tau,\tau':\mathbb{R}_{>0} \rightarrow \mathbb{R}_{>0}$, we define the function $\tau^{\min}_{\tau,\tau'}:\mathbb{R}_{>0} \rightarrow \mathbb{R}_{>0}$ by $\tau^{\min}_{\tau,\tau'}(n):=\min(\tau(n),\tau'(n))$, $\forall n \in \mathbb{R}_{>0}$.

\begin{lemma}\label{thm:dfr:prod}
	Consider groups $G=\langle S_G|R_G \rangle$ and $H=\langle S_H|R_H \rangle$, with $S_G$ and $S_H$ finite, and $S_G \cap S_H = \emptyset$. Let $R_{com}=\{[g,h]:g \in S_G,h \in S_H\}$. If $G$ has a $[k,d,\tau,\mathbb{A}]$-DFR and $H$ has a $[k',d',\tau',\mathbb{A}]$-DFR, then $G \times H=\langle S_G \sqcup S_H|R_G \cup R_H \cup R_{com} \rangle$ has a $[k+k',\max(d,d'),\tau^{\min}_{\tau,\tau'},\mathbb{A}]$-DFR. Moreover, if $G$ and $H$ have diagonal DFRs with the above parameters, then $G \times H$ has a diagonal DFR with the above parameters.
\end{lemma}
\begin{proof}
	By \Cref{thm:dfr:conv}(\ref{thm:dfr:conv:increaseDimension}), we may assume, without loss of generality, that $d'=d$ (i.e., we increase the smaller of $d,d'$ to $\max(d,d')$). Let $\mathcal{F}_G=\{\rho_1,\ldots,\rho_k\}$ be a $[k,d,\tau,\mathbb{A}]$-DFR for $G$ and $\mathcal{F}_H=\{\pi_1,\ldots,\pi_{k'}\}$ a $[k',d,\tau',\mathbb{A}]$-DFR for $H$. For each $j \in \{1,\ldots,k\}$, define a representation $\widehat{\rho}_j:G \times H \rightarrow \U(d)$ such that, $\widehat{\rho}_j(g,h)=\rho_j(g), \forall (g,h)\in G \times H$. Analogously, for each $j \in \{1,\ldots,k'\}$, we define a representation $\widehat{\pi}_j:G \times H \rightarrow \U(d)$ such that $\widehat{\pi}_j(g,h)=\pi_j(h), \forall (g,h)\in G \times H$. 
	
	Then $\mathcal{F}_{G \times H}=\{\widehat{\rho}_1,\ldots,\widehat{\rho}_k,\widehat{\pi}_1,\ldots,\widehat{\pi}_{k'}\}$ is the desired DFR. To see this, first notice that, $\forall (g,h) \in G \times H$, $l(g,h)= l(g)+l(h)$, where we write $l(g,h)$ in place of $l((g,h))$, to avoid cumbersome notation. By definition, $\tau$ and $\tau'$ are monotone non-increasing, and so, $\forall (g,h) \in G \times H$, we have $\tau(l(g,h)) \leq \tau(l(g))$ and $\tau'(l(g,h)) \leq \tau'(l(h))$. As $\mathcal{F}_G$ is a $[k,d,\tau,\mathbb{A}]$-DFR for $G$, we have that for each $g \in G_{\neq 1}$, $\exists j_g \in \{1,\ldots,k\}$ such that $\lvert \chi_{\rho_{j_g}}(g) \rvert \leq d-\tau(l(g))$. Analogously, for each $h \in H_{\neq 1_H}$, $\exists j_h \in \{1,\ldots,k'\}$ such that $\lvert \chi_{\pi_{j_h}}(h) \rvert \leq d-\tau(l(h))$. 
	
	Consider $(g,h) \in G \times H$, where $(g,h) \neq 1_{G \times H}=(1_G,1_H)$. Then we must have $g \neq 1_G$ or $h \neq 1_H$. If $g \neq 1_G$, then, by the above $\exists j_g$ such that $$\lvert \chi_{\widehat{\rho}_{j_g}}(g,h) \rvert = \lvert \chi_{\rho_{j_g}}(g) \rvert \leq d-\tau(l(g)) \leq d-\tau(l(g,h)).$$ If, $h \neq 1_H$, then, analogously, $\exists j_h$ such that $$\lvert \chi_{\widehat{\pi}_{j_h}}(g,h) \rvert = \lvert \chi_{\pi_{j_h}}(h) \rvert \leq d-\tau'(l(h)) \leq d-\tau'(l(g,h)).$$ Therefore, for any $(g,h) \in (G \times H)_{\neq 1_{G \times H}}$, there is some representation $\beta \in \mathcal{F}_{G \times H}$ for which \[\lvert \chi_{\beta}(g,h) \rvert \leq \max(d-\tau(l(g,h)),d-\tau'(l(g,h)))=d-\min(\tau(l(g,h)),\tau'(l(g,h)))=d-\tau^{\min}_{\tau,\tau'}(l(g,h)). \qedhere\] 
\end{proof}

Now, we show that a DFR of a group $G$ can be used to produce a DFR of a finitely generated subgroup of $G$, or of a finite-index overgroup of $G$.

\begin{lemma}\label{thm:dfr:subOver}
	Suppose $\mathcal{F}_G$ is a $[k,d,\tau,\mathbb{A}]$-DFR for a group $G=\langle S_G|R_G \rangle$, with $S_G$ finite. The following statements hold.
	\begin{enumerate}[(i)]
		\item \label{thm:dfr:subOver:sub} Suppose $H \leq G$, where $H=\langle S_H|R_H \rangle$, with $S_H$ finite. Then $\exists C \in \mathbb{R}_{>0}$ such that $H$ has a $[k,d,\tau \circ \eta_C,\mathbb{A}]$-DFR. If, moreover, $\mathcal{F}_G$ is a diagonal DFR, then $H$ will also have a diagonal DFR with the claimed parameters.
		\item \label{thm:dfr:subOver:over} Suppose $G \leq Q$, where $Q=\langle S_Q|R_Q \rangle$, with $S_Q$ finite, $S_G \subseteq S_Q$, and $r:=[Q:G]$ finite. Then $\exists C \in \mathbb{R}_{>0}$ such that $Q$ has a $[k,dr,\tau \circ \eta_C,\mathbb{A}]$-DFR.
	\end{enumerate}
\end{lemma}
\begin{proof}
	\begin{enumerate}[(i)]
		\item As $H \leq G$, $G$ admits a presentation $\langle S_G'|R_G' \rangle$ such that $S_G'$ is finite and $S_H \subseteq S_G'$. Writing $l_{S_H}(h)$ for the length of $h \in H$ relative to the generating set $S_H$ and $l_{S_G'}(g)$ for the length of $g \in G$ relative to the generating set $S_G'$, we immediately have that $l_{S_H}(h) \geq l_{S_G'}(h)$, $\forall h \in H \leq G$. By \Cref{thm:dfr:conv}(\ref{thm:dfr:conv:switchPresentation}), $\exists C \in \mathbb{R}_{>0}$ such that $\mathcal{F}_G$ is a $[k,d,\tau \circ \eta_C,\mathbb{A}]$-DFR of $G=\langle S_G'|R_G'\rangle$. Let $\tau'=\tau \circ \eta_C$ and let $\mathcal{F}_H=\{\pi_1,\ldots,\pi_k\}$, where $\pi_j=\Res^G_H(\rho_j)$. As $\mathcal{F}_G$ is a $[k,d,\tau',\mathbb{A}]$-DFR for $G$, we have that for each $h \in H \leq G$, where $h \neq 1_H=1_G$, $\exists j_h \in \{1,\ldots,k\}$ such that $\lvert \chi_{\rho_{j_h}}(h) \rvert \leq d-\tau'(l_{S_G'}(h))$. Notice that $\chi_{\pi_j}(h)=\chi_{\rho_j}(h)$, $\forall h \in H, \forall j \in \{1,\ldots,k\}$. As $\tau'$ is monotone non-increasing, $\tau'(l_{S_H}(h))\leq \tau'(l_{S_G'}(h))$. Therefore, $\forall h \in H_{\neq 1}$, $\exists j_h$ such that $$\lvert \chi_{\pi_{j_h}}(h) \rvert = \lvert \chi_{\rho_{j_h}}(h) \rvert \leq d-\tau'(l_{S_G'}(h)) \leq  d-\tau'(l_{S_H}(h)).$$ Therefore, $\mathcal{F}_H$ is the desired DFR for $H$. 
		
		\item For each $j \in \{1,\ldots,k\}$, let $\pi_j=\Ind_G^Q(\rho_j):Q \rightarrow \U(kr)$. Then $\mathcal{F}_Q=\{\pi_1,\ldots,\pi_k\}$ is the desired DFR. To see this, let $T \subseteq Q$ be a complete family of left coset representatives of $G$ in $Q$, where $1_Q \in T$. Notice that $\lvert T \rvert=[Q:G]=r$, with $r$ finite. Then, for any $q \in Q$, we have (see, for instance, \cite[Proposition 2.7.35]{kowalski2014introduction}) $$\chi_{\pi_{j}}(q)  =  \sum_{\substack{t \in T\\t^{-1} q t\in G}} \chi_{\rho_j}(t^{-1}qt).$$
		
		Let $l_Q(q)$ denote the length of $q \in Q$ relative to $S_Q$ and $l_G(g)$ denote the length of $g \in G \leq Q$ relative to $S_G$. Then $\exists C \in \mathbb{R}_{>0}$ such that $l_G(g) \leq Cl_Q(g)$, $\forall g \in G$, as $[Q:G]$ is finite. As $\tau$ is monotone non-increasing, $\tau(l_G(g)) \geq \tau(Cl_Q(g))$, $\forall g \in G$. Additionally, $\tau(l(g)) \leq d$, $\forall g \in G_{\neq 1}$. Therefore, if $g \in G_{\neq 1}$, then $d \geq \tau(l_G(g)) \geq \tau(Cl_Q(q))$.
		
		Fix $q \in Q_{\neq 1}$. First, suppose $q \in G$. As $\mathcal{F}_G=\{\rho_1,\ldots,\rho_k\}$ is a $[k,d,\tau,\mathbb{A}]$-DFR for $G$, we conclude that there is some $j$ such that $\lvert \chi_{\rho_{j}}(q) \rvert \leq d-\tau(l_G(q)) \leq d-\tau(Cl_Q(q))$. This immediately implies $$\lvert \chi_{\pi_{j}}(q) \rvert = \left\lvert \sum_{\substack{t \in T\\t^{-1} q t\in G}} \chi_{\rho_j}(t^{-1}qt)\right\rvert \leq \lvert \chi_{\rho_{j}}(1_Q^{-1} q 1_Q) \rvert + \left\lvert \sum_{\substack{t \in T\setminus{1_Q}\\t^{-1} q t\in G}} \chi_{\rho_j}(t^{-1}qt)\right\rvert \leq d-\tau(Cl_Q(q))+(r-1)d.$$ Therefore, there is some $j$ such that $\lvert \chi_{\pi_{j}}(q) \rvert \leq dr-\tau(Cl_Q(q))$, if $q \in G$. Next, suppose instead $q \not \in G$ and let $m=\lvert \{t \in T|t^{-1} q t \in G\}\rvert$. As $q \not \in G$, $1_Q^{-1} q 1_Q=q \not \in G$, and so $m \leq  \lvert T \rvert-1=r-1$. Therefore, $\forall j$, we have $$\lvert \chi_{\pi_j}(q) \rvert=\left\lvert \sum_{\substack{t \in T\\t^{-1} q t\in G}} \chi_{\rho_j}(t^{-1}qt)\right\rvert \leq  dm \leq dr-d \leq dr-\tau(Cl_Q(q)).$$ Therefore, $\forall q \in Q_{\neq 1}$, $\exists j$ such that $\lvert \chi_{\pi_j}(q) \rvert \leq dr-\tau(Cl_Q(q))$, as desired. \qedhere	
	\end{enumerate}
\end{proof}

\begin{remark}
	By the preceding lemma, any group $G$ that virtually has a DFR also has a DFR, but with worse parameters. As will be shown, it is possible to recognize $W_G$ using a DFR for a finite-index subgroup of $G$, thereby avoiding this worsening of parameters. 
\end{remark}

We now construct DFRs, with good parameters, for a wide class of groups. Recall that any finitely generated abelian group $G$ admits a unique decomposition $G \cong \mathbb{Z}^r \times \mathbb{Z}_{m_1} \times \cdots \times \mathbb{Z}_{m_t}$, where $m_i$ divides $m_{i+1}$, $\forall i \in \{1,\ldots,t-1\}$, and each $m_i \in \mathbb{N}_{\geq 2}$. Let $R(r,m_1,\ldots,m_t)= \{a_i^{m_i}:i \in \{1,\ldots,t\}\} \cup \{[a_i,a_j]:i,j \in \{1,\ldots,r+t\}\}$. 

\begin{lemma}\label{thm:dfr:abelian:finite}
	Consider the finite (hence finitely generated) abelian group $G = \mathbb{Z}_{m_1} \times \cdots \times \mathbb{Z}_{m_t}=\langle a_1,\ldots,a_t| R(0,m_1,\ldots,m_t) \rangle$. If $t=0$ (i.e., $G$ is the trivial group), then $G$ has a diagonal algebraic $[1,2,2]$-DFR. Otherwise, $G$ has a diagonal algebraic $\left[t,2,\frac{19 \pi^2}{24 m_t^2}\right]$-DFR. 
\end{lemma}
\begin{proof}
	If $t=0$, the claim is obvious. Suppose $t>0$. By \Cref{thm:dfr:init:Z/mZ}, each factor $\mathbb{Z}_{m_i}=\langle a|a^{m_i}\rangle$ has a diagonal algebraic $\left[1,2,\frac{19 \pi^2}{24 m_i^2}\right]$-DFR. Notice that $m_1 \leq \cdots \leq m_t$, as each $m_i$ divides $m_{i+1}$. The existence of the desired DFR follows from \Cref{thm:dfr:prod}.
\end{proof}

\begin{theorem}\label{thm:dfr:abelian} $\exists C_1 \in \mathbb{R}_{>0}$ such that, for any finitely generated abelian group $G = \mathbb{Z}^r \times \mathbb{Z}_{m_1} \times \cdots \times \mathbb{Z}_{m_t}=\langle a_1,\ldots,a_{r+t}| R(r,m_1,\ldots,m_t) \rangle$, the following statements hold.
	\begin{enumerate}[(i)]
		\item \label{thm:dfr:abelian:alg} $\exists C_2 \in \mathbb{R}_{>0}$ such that $G$ has a diagonal algebraic $\left[r+t,2,C_2 n^{-C_1}\right]$-DFR.
		\item \label{thm:dfr:abelian:nonAlg} $\forall \delta \in \mathbb{R}_{>0}$, $\exists C_3 \in \mathbb{R}_{>0}$, such that $G$ has a diagonal $\left[r \left(1+ \lfloor \frac{2}{\delta} \rfloor\right)+t,2,C_3 n^{-\delta},\widetilde{\mathbb{C}}\right]$-DFR.		
	\end{enumerate}
\end{theorem} 
\begin{proof}
	By \Cref{thm:dfr:init:Z:alg}, $\exists D_1,D_2 \in \mathbb{R}_{>0}$ such that $\mathbb{Z}$ has a diagonal algebraic $[1,2,D_2 n^{-D_1}]$-DFR, which we call $\mathcal{F}$. We set $C_1=D_1$. Let $H_1= \mathbb{Z}^r$ and $H_2=\mathbb{Z}_{m_1} \times \cdots \times \mathbb{Z}_{m_t}$. If $r=0$, both claims follow trivially from \Cref{thm:dfr:abelian:finite}. Suppose $r>0$.
	\begin{enumerate}[(i)]	
		\item Using the DFR $\mathcal{F}$ of $\mathbb{Z}$, \Cref{thm:dfr:prod} implies $H_1$ has a diagonal algebraic $[r,2,D_2 n^{-C_1}]$-DFR $\mathcal{H}_1$. If $t=0$, then $G=H_1$; therefore, $\mathcal{H}_1$ is the desired DFR for $G$, with $C_2=D_2$, and we are done. If $t>0$, \Cref{thm:dfr:abelian:finite} implies $H_2$ has a diagonal algebraic $\left[1,2,\frac{19 \pi^2}{24 m_t^2}\right]$-DFR $\mathcal{H}_2$. Set $C_2=\min(D_2,\frac{19 \pi^2}{24 m_t^2})$. By \Cref{thm:dfr:prod}, we conclude $G=H_1 \times H_2$ has a DFR with the claimed parameters.
		\item By \Cref{thm:dfr:init:Z:nonAlg}, $\exists D \in \mathbb{R}_{>0}$ such that $\mathbb{Z}$ has a diagonal $\left[1+\lfloor \frac{2}{\delta} \rfloor,2,Dn^{-\delta},\widetilde{\mathbb{C}}\right]$-DFR, $\mathcal{F}'$. The remainder of the proof is analogous to that of part (i), using $\mathcal{F}'$ in place of $\mathcal{F}$. \qedhere		
	\end{enumerate}
\end{proof}

As in \Cref{sec:intro:results}, $\widehat{\Pi}_1$ denotes the set of all finitely generated virtually abelian groups. For $G \in \widehat{\Pi}_1$, there is a unique $r \in \mathbb{N}$ such that $G$ is virtually $\mathbb{Z}^r$. We have the following corollary. 

\begin{corollary}\label{thm:dfr:virtAbelian}
	$\exists C \in \mathbb{R}_{>0}$ such that, $\forall G \in \widehat{\Pi}_1$, the following holds.
	\begin{enumerate}[(i)]
		\item \label{thm:dfr:virtAbelian:alg} $\exists D \in \mathbb{R}_{>0},\exists K \in \mathbb{N}_{\geq 1}$, such that $G$ virtually has a diagonal algebraic $[K,2,D n^{-C}]$-DFR.
		\item \label{thm:dfr:virtAbelian:nonAlg} $\forall \delta \in \mathbb{R}_{>0}$, $\exists D \in \mathbb{R}_{>0}, \exists K \in \mathbb{N}_{\geq 1}$, $G$ virtually has a diagonal $\left[K,2,D n^{-\delta},\widetilde{\mathbb{C}}\right]$-DFR.
	\end{enumerate}
\end{corollary}

Next, we consider groups that can be built from finitely generated free groups. 

\begin{lemma}\label{thm:dfr:freeSingle}
	$\forall r \in \mathbb{N}$, $\exists C \in \mathbb{R}_{\geq 1}$, $F_r=\langle a_1,\ldots,a_r| \rangle$ has an algebraic $[1,2,C^{-n}]$-DFR.
\end{lemma}
\begin{proof}
	As $F_0=\{1\}$ and $F_1 = \mathbb{Z}$, \Cref{thm:dfr:abelian} immediately implies the claim when $r \in \{0,1\}$. Next, consider the case in which $r=2$. By \Cref{thm:dfr:init:F2}, $\exists C \in \mathbb{R}_{\geq 1}$ such that  $F_2=\langle a_1,a_2| \rangle$ has an algebraic $[1,2,C^{-n}]$-DFR. Finally, suppose $r>2$. By the Nielsen-Schreier theorem, $F_2$ has a finite-index subgroup isomorphic to $F_r$; the claim immediately follows from \Cref{thm:dfr:subOver}(\ref{thm:dfr:subOver:sub}).
\end{proof}

\begin{theorem}\label{thm:dfr:free}
	Suppose $G=\langle S|R \rangle$, with $S$ finite, such that $G \leq F_{r_1} \times \cdots \times F_{r_t}$, for some $r_1,\ldots,r_t \in \mathbb{N}$. Then $\exists C \in \mathbb{R}_{\geq 1}$ such that $G$ has an algebraic $[t,2,C^{-n}]$-DFR. 
\end{theorem}
\begin{proof}	
	By \Cref{thm:dfr:freeSingle}, each $F_{r_i}$ has an algebraic $[1,2,C_i^{-n}]$-DFR, for some $C_i \in \mathbb{R}_{\geq 1}$. \Cref{thm:dfr:prod} implies that $F_{r_1} \times \cdots \times F_{r_t}$ has an algebraic $[t,2,C^{-n}]$-DFR, where $C=\max_i C_i$, and \Cref{thm:dfr:subOver}(\ref{thm:dfr:subOver:sub}) then implies $G$ has a DFR with the claimed parameters.
\end{proof}

As in \Cref{sec:intro:results}, $\widehat{\Pi}_2$ denotes the class of finitely generated groups that are virtually a subgroup of a direct product of finitely-many finite-rank free groups.

\begin{corollary}\label{thm:dfr:virtFree}
	$\forall G \in \widehat{\Pi}_2, \exists K \in \mathbb{N}_{\geq 1}, \exists C \in \mathbb{R}_{\geq 1}$, such that $G$ virtually has an algebraic $[K,2,C^{-n}]$-DFR.
\end{corollary}

We conclude with a ``generic'' construction that covers all groups that have algebraic DFRs. We remark that while this does partially subsume all other results in this section, it does not do so completely, as the earlier constructions of DFRs, for certain particular groups, yield better parameters. 

\begin{theorem}\label{thm:dfr:embed}
	Consider a group $G=\langle S| R \rangle$, with $S$ finite, where $G$ is not the trivial group. Suppose $G$ has a faithful representation $\pi:G \rightarrow \U(l,\overline{\mathbb{Q}})$. Then $\pi$ has a (unique, up to isomorphism) set of irreducible subrepresentations $\{\pi_j:G \rightarrow \U(d_j,\overline{\mathbb{Q}})\}_{j=1}^m$ such that $\pi \cong \pi_1 \oplus \cdots \oplus \pi_m$.	Let $d_{\max}=\max_j d_j$. Define the value $d$ as follows: if $\bigcap_j \Pker(\pi_j)=\{1_G\}$, let $d=d_{\max}$, otherwise, let $d=d_{\max}+1$. Partition the non-trivial $\pi_j$ into isomorphism classes (i.e., only consider those $\pi_j$ which are not the trivial representation; $\pi_{j_1}$ and $\pi_{j_2}$ belong to the same isomorphism class if $\pi_{j_1} \cong \pi_{j_2}$) and let $k$ denote the number of isomorphism classes that appear. Then $\exists C \in \mathbb{R}_{\geq 1}$ such that $G$ has an algebraic $[k,d,C^{-n}]$-DFR. 
\end{theorem}
\begin{proof}
	Notice that, as $G$ is not the trivial group, $d \geq 2$. Assume that the $\pi_j$ are ordered such that $\pi_1,\ldots,\pi_k$ are representatives of the $k$ distinct isomorphism classes of the non-trivial representations that appear among the $\pi_j$. For each $j \in \{1,\ldots,k\}$, define the representation $\rho_j=\pi_j \oplus \mathbf{1}_{d-d_j}:G \rightarrow \U(d,\overline{\mathbb{Q}})$. By \Cref{thm:characterAlgDiophantine}, $\forall j \in \{1,\ldots,k\}, \exists C_j \in \mathbb{R}_{\geq 1}$ such that, $\forall g \not \in \Pker(\rho_j)$, $\lvert \chi_{\rho_j}(g) \rvert \leq d - C_j^{-l(g)}$. Set $C=\max_j C_j$.
	
	Next, notice that $\bigcap_j \Pker(\rho_j)=\{1_G\}$. If $\bigcap_j \Pker(\pi_j)=\{1_G\}$, then this is obvious. Suppose $\bigcap_j \Pker(\pi_j) \neq \{1_G\}$. Then $d=d_{\max}+1>d_j$, $\forall j$, which implies $\rho_j=\pi_j \oplus \mathbf{1}_{t_j}$, where $t_j:=d-d_j \geq 1$. Therefore, for each $j$, $\rho_j(G) \cap Z(\U(d,\overline{\mathbb{Q}})) = I_d$, and so, by definition, $\Pker(\rho_j)=\ker(\rho_j)$.  As $\pi$ is faithful, $$\{1_G\}=\bigcap_{j=1}^m \ker(\pi_j)=\bigcap_{j=1}^k \ker(\rho_j)=\bigcap_{j=1}^k \Pker(\rho_j).$$ 
	
	Thus, $\forall g \in G_{\neq 1}$, $\exists j$ such that $g \not \in \Pker(\rho_j)$, which implies $\lvert \chi_{\rho_j}(g) \rvert \leq d - C_j^{-l(g)} \leq d - C^{-l(g)}$. Therefore, $\{\rho_1,\ldots,\rho_k\}$ is an algebraic $[k,d,C^{-n}]$-DFR for $G$. 
\end{proof}

\subsection{Projective DFRs}\label{sec:dfr:projective}

A DFR $\mathcal{F}=\{\rho_1,\ldots,\rho_j\}$ of a group $G$ is a set of unitary representations of $G$, i.e., group homomorphisms $\rho_j:G \rightarrow \U(d)$. We next consider a slight generalization. A \textit{projective} unitary representation of $G$ is a group homomorphism $\rho:G \rightarrow \PU(d)=\U(d)/Z(\U(d))$. We may (non-uniquely) lift any such $\rho$ to a function $\widehat{\rho}:G \rightarrow \U(d)$ (i.e., $\gamma \circ \widehat{\rho}=\rho$, where $\gamma:\U(d) \rightarrow \PU(d)$ is the canonical projection). Note that $\widehat{\rho}$ is not necessarily a group homomorphism and that certain projective representations $\rho$ cannot be lifted to an ordinary representation. However, for any two lifts, $\widehat{\rho}_1$ and $\widehat{\rho}_2$, of $\rho$, we have $\lvert \chi_{\hat{\rho}_1}(g) \rvert=\lvert \chi_{\hat{\rho}_2}(g) \rvert$, $\forall g \in G$. Therefore, the function $\lvert \chi_{\rho}(\cdot) \rvert:G \rightarrow \mathbb{R}$ given by $\lvert \chi_{\rho}(g) \rvert=\lvert \chi_{\hat{\rho}}(g)\rvert$ is well-defined. 

We then define a $[k,d,\tau,\mathbb{A}]$-PDFR as a set of projective representations $\mathcal{F}=\{\rho_1,\ldots,\rho_j\}$ that satisfies \Cref{def:dfr} where ``representation'' is replaced by ``projective representation'' in that definition. As we will observe in the following section, the same process that allows a DFR for a group $G$ to be used to produce a 2QCFA for the word problem $W_G$, can also be applied to a PDFR. If a PDFR consists entirely of representations into $\PU(d,\overline{\mathbb{Q}})=\U(d,\overline{\mathbb{Q}})/Z(\U(d,\overline{\mathbb{Q}}))$, we say it is an \textit{algebraic} PDFR. The following variant of \Cref{thm:dfr:embed} follows by a precisely analogous proof.

\begin{theorem}\label{thm:pdfr:generic}
	Suppose the group $G= \langle S|R \rangle$, with $S$ finite, has a family $\mathcal{F}=\{\rho_1,\ldots,\rho_k\}$ of projective representations $\rho_j:G \rightarrow \PU(d,\overline{\mathbb{Q}})$, such that $\bigcap_j \ker(\rho_j)=\{1_G\}$. Then $\exists C \in \mathbb{R}_{\geq 1}$ such that $\mathcal{F}$ is an algebraic $[k,d,C^{-n}]$-PDFR for $G$.
\end{theorem}

\subsection{Unbounded-Error DFRs}\label{sec:dfr:unbounded}

If $\mathcal{F}=\{\rho_1,\ldots,\rho_k\}$ is a DFR for a group $G$, then $\bigcap_j \Pker(\rho_j)=\{1_G\}$. However, a crucial element in the definition of a DFR is the requirement that, much more strongly, all $g \in G_{\neq 1}$ are ``far'' from being in $\bigcap_j \Pker(\rho_j)$; in particular, if $\mathcal{F}$ is a $[k,d,\tau,\mathbb{A}]$-DFR, then  $\forall g \in G_{\neq 1}, \exists j$ such that $\lvert \chi_{\rho_j}(g)\rvert \leq d-\tau(l(g))$. This requirement is essential in order for our construction of a 2QCFA, that recognizes $W_G$ using a DFR for $G$, to operate with \textit{bounded} error. We next consider a generalization of a DFR, where this requirement is removed, which will then yield a 2QCFA that recognizes $W_G$ with \textit{unbounded} error.

We say $\mathcal{F}=\{\rho_1,\ldots,\rho_k\}$ is an \textit{unbounded-error} $[k,d,\mathbb{A}]$-DFR for a group $G=\langle S|R \rangle$ if the conditions of \Cref{def:dfr} hold, where \Cref{def:dfr}(\ref{def:dfr:distinguish}) is replaced by  \Cref{def:dfr}(\ref{def:dfr:distinguish})': $\forall g \in G_{\neq 1}$, $\exists j$ such that $\lvert \chi_{\rho_j}(g) \rvert < d$. This condition is equivalent to $\bigcap_j \Pker(\rho_j)=\{1_G\}$.

Note that, by \Cref{thm:characterAlgDiophantine}, any algebraic unbounded-error $[k,d]$-DFR is also an algebraic $[k,d,C^{-n}]$-DFR, for some $C \in \mathbb{R}_{\geq 1}$; furthermore, as noted in the discussion following \Cref{def:dfrTypes}, only a finitely generated abelian group could have a diagonal unbounded-error $[k,d]$-DFR, and all finitely generated abelian groups were shown to have DFRs in \Cref{thm:dfr:abelian}. Therefore, in order to obtain something new, we must consider unbounded-error DFRs that are neither algebraic nor diagonal.

We will show that any $G \in \widehat{\Pi}_3$ has an unbounded-error DFR. We begin by again considering the group $\mathbb{Z}^r$, for $r \in \mathbb{N}_{\geq 1}$. While the DFRs produced by \Cref{thm:dfr:abelian} suffice for establishing all of our results concerning the recognizability of the word problem for $\mathbb{Z}^r$, we next exhibit a different construction of a DFR for $\mathbb{Z}^r$, which we will require in order to exhibit an unbounded-error DFR of a related group. In the following, for a commutative (unital) ring $R$, let $\SO(2,R)$ denote the group of $2 \times 2$ orthogonal matrices of determinant $1$ whose entries lie in $R$. For a set of prime numbers $\mathcal{P}=\{p_1,\ldots,p_m\}$, let $\mathbb{Z}[\frac{1}{p_1},\ldots,\frac{1}{p_m}]$ denote the ring obtained by adjoining $\frac{1}{p_1},\ldots,\frac{1}{p_m}$ to the ring $\mathbb{Z}$, i.e., $\mathbb{Z}[\frac{1}{p_1},\ldots,\frac{1}{p_m}]$ is the localization of $\mathbb{Z}$ away from $\mathcal{P}$. Notice that $\SO(2,\mathbb{Z}[\frac{1}{p_1},\ldots,\frac{1}{p_m}]) \leq \SO(2,\mathbb{Q}) \leq \SU(2,\mathbb{Q}) \leq \SU(2,\overline{\mathbb{Q}})$.

\begin{lemma}\label{thm:dfr:ratCircleFreeAbelian}
	Consider the group $\mathbb{Z}^r=\langle S_r|R_r \rangle$, where $S_r=\{a_1,\ldots,a_r\}$ and $R_r=\{[a_i,a_j]|i,j \in \{1,\ldots,r\}\}$. There is a representation $\rho:\mathbb{Z}^r \rightarrow \SO(2,\mathbb{Z}[\frac{1}{p_1},\ldots,\frac{1}{p_r}])$ and $D_1,D_2 \in \mathbb{R}_{>0}$, such that $\{\rho\}$ is a $[1,2,D_2 n^{-D_1}]$-algebraic DFR for $\mathbb{Z}^r$.
\end{lemma}
\begin{proof}
	Fundamentally, we follow the construction of Tan \cite{tan1996group} of the rational points on the unit circle. Let $p_j$ denote the $j^{\text{th}}$ prime number that is congruent to $1$ modulo $4$, and let $m_j,n_j \in \mathbb{N}$ denote the (unique) values which satisfy $p_j=m_j^2+n_j^2$ and $m_j > n_j > 0$. Define the representation $\rho: \mathbb{Z}^r \rightarrow \SO(2,\mathbb{Z}[\frac{1}{p_1},\ldots,\frac{1}{p_r}])$ such that $$a_j \mapsto \frac{1}{p_j} \begin{pmatrix} m_j^2-n_j^2 & 2m_j n_j \\ -2m_jn_j & m_j^2-n_j^2\end{pmatrix}, \ \ \forall j \in \{1,\ldots,r\}.$$ 
	
	Notice that $\rho(a_j)$ has eigenvalues $p_j^{-1}(m_j^2-n_j^2 \pm 2 m_j n_j i)$. As $\SO(2,\mathbb{Z}[\frac{1}{p_1},\ldots,\frac{1}{p_r}])$ is abelian, the $\rho(a_j)$ are simultaneously diagonalizable. Define $Y \in \U(2)$ such that, $\forall j$, $Y \rho(a_j) Y^{-1}=D_j$, where $D_j$ is a $2 \times 2$ diagonal matrix whose diagonal entries are the eigenvalues $p_j^{-1}(m_j^2-n_j^2 \pm 2 m_j n_j i)$. Define $\alpha_j \in (\mathbb{R} \cap (- \pi,\pi))$ such that $D_j=\text{diag}[e^{i \alpha_j},e^{-i \alpha_j}]$. 
	
	For some $(q_1,\ldots,q_r) \in \mathbb{Z}^r$, consider the element $g=a_1^{q_1} \cdots a_r^{q_r} \in \mathbb{Z}^r$. Then $$\chi_{\rho}(g)=\Tr\left(\prod_{j=1}^r \rho(a_j)^{q_j}\right)=\Tr\left(\prod_{j=1}^r \left(Y\rho(a_j)Y^{-1}\right)^{q_j}\right)=2 \cos\left(\sum_j q_j\alpha_j\right).$$
	
	Let $L=\{\beta \in \mathbb{C}_{\neq 0}|e^{\beta} \in \overline{\mathbb{Q}}\}$. Let $\beta_0 = i\pi$ and, for $j \in \{1,\ldots,r\}$, let $\beta_j=i\alpha_j$. Then $\beta_0,\ldots,\beta_r \in L$. By \cite[Theorem 1]{tan1996group}, $\rho$ is P-faithful, which immediately implies $\beta_0,\ldots,\beta_r$ are linearly independent over $\mathbb{Q}$. By Proposition~\ref{thm:baker}, $\exists C \in \mathbb{R}_{>0}$ such that, $\forall (q_0,\ldots,q_r) \in \mathbb{Z}^{r+1}$, where $q_{\max}:=\max_j \lvert q_j \rvert>0$, we have $\lvert \sum_j q_j \beta_j \rvert \geq (eq_{\max})^{-C}$. 
	
	Consider any $g=a_1^{q_1} \cdots a_r^{q_r} \in \mathbb{Z}^r_{\neq 1_{\mathbb{Z}^r}}$ (i.e., not all $q_i=0$). Let $q_0=\text{round}(\frac{1}{\pi} \sum_{j=1}^r q_j \alpha_j)$ and observe that, by construction $\lvert \alpha_j \rvert \leq \pi$, $\forall j$, and so $\lvert q_0 \rvert \leq \sum_{j=1}^r \lvert q_j \rvert=l(g)$. Therefore, $q_{\max}:=\max_{j \in \{0,\ldots,r\}} q_j \leq l(g)$, which implies $$\min_{m \in \mathbb{Z}} \left\lvert m \pi +\sum_{j=1}^r q_j \alpha_j \right\rvert = \left\lvert q_0 \pi +\sum_{j=1}^r q_j \alpha_j \right\rvert = \left\lvert q_0 \beta_0 +\sum_{j=1}^r q_j \beta_j \right\rvert \geq (el(g))^{-C}.$$ Therefore, $$\lvert\chi_{\rho}(g)\rvert=2\left\lvert \cos\left(\sum_j q_j\alpha_j\right)\right\rvert \leq 2-C' \min_{m \in \mathbb{Z}} \left\lvert m \pi +\sum_{j=1}^r q_j \alpha_j \right\rvert^2 \leq 2-C' (el(g))^{-2C},$$ for a constant $C' \in \mathbb{R}_{>0}$. We then conclude that $\{\rho\}$ is a $[1,2,D_2 n^{-D_1}]$-algebraic DFR for $\mathbb{Z}^r$, where $D_1=2C$ and $D_2=C'e^{-2C}$.
\end{proof}   

\begin{lemma}\label{thm:unboundedDFR:ZfreeZr}
	For any $r \in \mathbb{N}_{\geq 1}$, $\mathbb{Z} * \mathbb{Z}^r$ has an unbounded-error $[1,2,\widetilde{\mathbb{C}}]$-DFR. 
\end{lemma}
\begin{proof}
	Fix $r$. Let $S_r=\{x_1,\ldots,x_r\}$ and let $R_r=\{[x_i,x_j]|i,j \in \{1,\ldots,r\}\}$. By Lemma~\ref{thm:dfr:ratCircleFreeAbelian}, the group $A:=\mathbb{Z}^r=\langle S_r|R_r \rangle$ has a P-faithful representation $\rho:A \rightarrow \SU(2,\mathbb{Q})$, and the group $B:=\mathbb{Z}=\langle \{y\}|\rangle$ has a P-faithful representation $\pi:B \rightarrow \SU(2,\mathbb{Q})$. Notice that, $\forall a \in A_{\neq 1_A}$ both off-diagonal entries of the matrix $\rho(a)$ are nonzero. To see this, consider some $a \in A_{\neq 1_A}$. As $\rho(a) \in \SU(2)$, its two off-diagonal entries are equal in magnitude, and so they are both zero or both nonzero. If they are both zero, then $\rho(a)$ is diagonal; however, the only diagonal matrices in $\SU(2,\mathbb{Q})$ are $\{\pm I_2\}$, which would then imply $\rho(a) \in \{\pm I_2\}=Z(\SU(2))$, which contradicts the fact that $\rho$ is P-faithful. By a symmetric argument, $\forall b \in B_{\neq 1_B}$, both off-diagonal entries of the matrix $\pi(b)$ are nonzero.
	
	We now fundamentally follow (the proof of) Shalen \cite[Proposition 1.3]{shalen1979linear} to produce a P-faithful representation of $A*B \cong \mathbb{Z} * \mathbb{Z}^r$. Fix $\alpha \in ((\mathbb{R} \cap \overline{\mathbb{Q}})\setminus \mathbb{Q})$, let $\lambda=e^{\pi i \alpha}$, and notice that, by the Gel'fond-Schneider theorem, $\lambda \not \in \overline{\mathbb{Q}}$. Let $\Lambda=\text{diag}[\lambda,\lambda^2]$, the $2 \times 2$ diagonal matrix with diagonal entries $\lambda$ and $\lambda^2$, and observe that $\Lambda \in \T(2,\widetilde{\mathbb{C}})$. Define the representation $\widehat{\rho}:A \rightarrow \SU(2)$ by $\widehat{\rho}(a)=\Lambda \rho(a) \Lambda^{-1}$, $\forall a \in A$. Define the representation $\gamma:A*B \rightarrow \SU(2)$ such that $\gamma(a)=\widehat{\rho}(a)$, $\forall a \in A$ and $\gamma(b)=\pi(b)$, $\forall b \in B$ (where $\gamma$ is uniquely defined by the universal property of the free product). By Shalen \cite[Proposition 1.3]{shalen1979linear}, $\gamma$ is a P-faithful representation. Moreover, $\pi(y) \in \SU(2,\mathbb{Q}) \leq \U(2,\overline{\mathbb{Q}})$, and for each $x_j \in S_r$, $\widehat{\rho}(x_j)=\Lambda \rho(x_j) \Lambda^{-1}$, and so $\widehat{\rho}(x_j)$ is the product of three matrices in $\U(2,\overline{\mathbb{Q}}) \cup \T(2,\widetilde{\mathbb{C}})$. As $\{y\} \sqcup S_r$ is a generating set for $A*B$, this implies that the image of each such generator under $\gamma$ is expressible as the product of at most three matrices in $\U(2,\overline{\mathbb{Q}}) \cup \T(2,\widetilde{\mathbb{C}})$. Therefore, $\{\gamma\}$ is an unbounded-error $[1,2,\widetilde{\mathbb{C}}]$-DFR for $A*B \cong \mathbb{Z} * \mathbb{Z}^r$.
\end{proof}

\begin{theorem}\label{thm:unboundedDFR:pi3}
	$\forall G \in \widehat{\Pi}_3, \exists k \in \mathbb{N}$ such that $G$ virtually has an unbounded-error $[k,2,\widetilde{\mathbb{C}}]$-DFR.
\end{theorem}
\begin{proof}
	Consider a group $H \in \Sigma_2$. Such an $H$ is of the form $H \cong \mathbb{Z}^{r_1} * \cdots * \mathbb{Z}^{r_m}$, for some $r_1,\ldots,r_m \in \mathbb{N}$. Let $r=\max_j r_j$. Then, by a straightforward application of the Kurosh subgroup theorem, $H$ embeds in $\mathbb{Z} * \mathbb{Z}^r$, which implies $H$ has an unbounded-error $[1,2,\widetilde{\mathbb{C}}]$-DFR, by \Cref{thm:unboundedDFR:ZfreeZr}. Next, consider a group $L \in \Pi_3$; such a group is of the form $L \cong H_1 \times \cdots \times H_k$, for some $H_1,\ldots,H_k \in \Sigma_2$. As all such $H_j$ have unbounded-error $[1,2,\widetilde{\mathbb{C}}]$-DFRs, we conclude, by an argument identical to that of \Cref{thm:dfr:prod}, that $L$ has an unbounded-error $[k,2,\widetilde{\mathbb{C}}]$-DFR. Finally, for any $G \in \widehat{\Pi}_3$, $G$ has a finitely-index subgroup $K$ such that $K$ is isomorphic to a finitely generated subgroup of some $L \in \Pi_3$. As just observed, any such $L$ has an unbounded-error $[k,2]$-DFR, for some $k$, and so, by the same argument as in \Cref{thm:dfr:subOver}(\ref{thm:dfr:subOver:sub}), $K$ has an unbounded-error $[k,2,\widetilde{\mathbb{C}}]$-DFR. We then conclude $G$ virtually has an unbounded-error $[k,2,\widetilde{\mathbb{C}}]$-DFR, as desired.
\end{proof}

\section{Recognizing the Word Problem of a Group with a 2QCFA}\label{sec:2QCFAword}

In this section, we use a DFR for a group $G$ to construct a 2QCFA that recognizes the word problem of $G$, as well as for certain other groups related to $G$. 

\begin{definition2}\label{def:goodSetForDFR}
	Consider a group $G=\langle S|R \rangle$, with $S$ finite. As before, let $\Sigma=S \sqcup S^{-1}$, let $\phi:\Sigma^* \rightarrow G$ denote the natural map that takes each string in $\Sigma^*$ to the element of $G$ that it represents, and let $W_G:=W_{G=\langle S|R \rangle}=\{w \in \Sigma^*:\phi(w)=1_G\}$ denote the word problem of $G$ with respect to the given presentation. Suppose $\mathcal{F}=\{\rho_1,\ldots,\rho_k\}$ is a $[k,d,\tau,\mathbb{A}]$-DFR (or PDFR) for $G$. By \Cref{thm:dfr:intersectPKer}, if $w \in W_G$, then $\lvert \chi_{\rho_j}(\phi(w)) \rvert = d$, $\forall j$, and if $w \not \in W_G$, then $\exists j$ where $\lvert \chi_{\rho_j}(\phi(w)) \rvert \leq d - \tau(l(\phi(w)))$. Let $G_j=\{g \in G : \lvert \chi_{\rho_j}(g) \rvert \leq d-\tau(l(g))\}$. A 2QCFA can recognize $W_G$ by checking if $\phi(w) \in \bigcup_j G_j=G_{\neq 1}$.
\end{definition2}

The well-known Hadamard test may be used to estimate $\chi_{\rho_j}(\phi(w))=\Tr(\rho_j(\phi(w)))$; however, as we wish to produce a 2QCFA that has as few quantum states as possible, we wish to avoid the use of ancilla, and so we follow a slightly different approach.

\subsection{Computing with DFRs}\label{sec:2QCFAword:usingDFRs}

We begin by defining several useful 2QCFA subroutines.

\begin{definition2}\label{def:2qcfaRound}
	Suppose $M$ is a 2QCFA with $d \geq 2$ quantum basis states $Q=\{q_1,\ldots,q_d\}$, quantum start state $q_1 \in Q$, and alphabet $\Sigma$. 
	\begin{enumerate}[(a)]
		\item \label{def:2qcfaRound:pairTransition} Suppose  $\ket{\psi_1}=\sum_q \alpha_q \ket{q}$ and $\ket{\psi_2}=\sum_q \beta_q \ket{q}$, where $\alpha_q,\beta_q \in \overline{\mathbb{Q}}, \forall q \in Q$. There are (many) $t \in \U(d,\overline{\mathbb{Q}})$ such that $t \ket{\psi_1} = \ket{\psi_2}$. Let $\mathcal{T}_{\ket{\psi_1} \rightarrow \ket{\psi_2}}$ denote an arbitrary such $t$.
		\item \label{def:2qcfaRound:unitaryRound}  Let $\pi:G \rightarrow \U(d)$ be a representation of $G$ and let $\ket{\psi}=\sum_q \beta_q \ket{q}$, where $\beta_q \in \overline{\mathbb{Q}}$, $\forall q \in Q$. Then the \textit{unitary round} $\mathcal{U}(\pi,\ket{\psi})$ is a particular sub-computation of $M$ on $w$, defined as follows. The round begins with the quantum register in the superposition $\ket{q_1}$ and the tape head at the right end of the tape. On reading $\#_R$, $M$ performs the unitary transformation $\mathcal{T}_{\ket{q_1} \rightarrow \ket{\psi}}$ to its quantum register, and moves its head to the left. On reading a symbol $\sigma \in \Sigma$, $M$ performs the unitary transformation $\pi(\phi(\sigma))$ to the quantum register and moves its head left. When the tape head first reaches the left end of the tape (i.e., the first time the symbol $\#_L$ is read), $M$ performs the identity transformation to its quantum register, and does not move its head, at which point the round ends. As $\phi$ is a (monoid) homomorphism and $\pi$ is a (group) homomorphism, we immediately conclude that, at the end of the round, the quantum register is in the superposition $\pi(\phi(w))\ket{\psi}$.
		\item \label{def:2qcfaRound:measureRound:basic} For $t \in \U(d)$, a \textit{measurement round} $\mathcal{M}(\pi,\ket{\psi},t)$ is a sub-computation of $M$ that begins with the unitary round $\mathcal{U}(\pi,\ket{\psi})$. Then $M$ performs the unitary transformation $t$, and does not move its head. After which $M$ performs the quantum measurement specified by the partition $B=\{B_0,B_1\}$ of $Q$ given by $B_0=\{q_2,\ldots,q_d\}$ and $B_1=\{q_1\}$, producing some \textit{result} $r \in \{0,1\}$; then $M$ records $r$ in its classical state, and does not move its head, at which point the round is over. 
	\end{enumerate} 
\end{definition2}

\begin{lemma}\label{thm:algorithm:measureRep:diag}
	Using the notation of \Cref{def:goodSetForDFR}, let $\ket{1}=\frac{1}{\sqrt{d}}\sum_j \ket{q_j}$. Fix any $F \in \U(d,\overline{\mathbb{Q}})$ such that all entries in the first row of $F$ are equal to $\frac{1}{\sqrt{d}}$. For concreteness, we take $F$ as the usual (unitary) $d \times d$ DFT matrix, i.e., the $(u,v)$ entry of $F$ is given by $F[u,v]=\frac{1}{\sqrt{d}} e^{-\frac{2 \pi i}{d}(u-1)(v-1)}$, $\forall u,v \in \{1,\ldots,d\}$. Then, $\forall w \in \Sigma^*, \forall j \in \{1,\ldots,k\}$, the result $r$ of the measurement round $\mathcal{M}(\rho_j,\ket{1},F)$ (on input $w$) has the following properties.
	\begin{enumerate}[(a)]
		\item \label{thm:algorithm:measureRep:diag:completeness} (Perfect Completeness) If $\phi(w)=1_G$, then $\Pr[r=1]=1$.
		\item \label{thm:algorithm:measureRep:diag:soundness} (Soundness) If $\phi(w) \in G_j$, then $\Pr[r=0] \geq \frac{\tau(n)}{d}-\delta$, where $\delta=\max_v \lvert \sum_{u \neq v} \rho_j(\phi(w))[u,v] \rvert$. If, moreover, $\mathcal{F}$ is a diagonal DFR, then $\Pr[r=0] \geq \frac{\tau(n)}{d}$.
	\end{enumerate}
\end{lemma}
\begin{proof}
	Notice that, for any $M \in \U(d)$, $FM \ket{1}=\left(\frac{1}{d}\sum_{u,v}M[u,v] \right) \ket{q_1}+\sum_{h> 1} \alpha_h \ket{q_h}$, for some $\alpha_2,\ldots,\alpha_d \in \mathbb{C}$. Therefore, $\Pr[r=1]=  \lvert\frac{1}{d} \sum_{u,v}\rho_j(\phi(w))[u,v] \rvert^2$. If $\phi(w)=1_G$, then $\rho_j(\phi(w))=I_d$, where $I_d$ denotes the $d \times d$ identity matrix; therefore, $\Pr[r=1]=1$, as desired. If $\phi(w) \in G_j$, then $$\Pr[r=0]=1-\left\lvert \frac{1}{d}\sum_{u,v} \rho_j(\phi(w))[u,v]\right\rvert^2 \geq 1-\frac{1}{d^2}\left(\left\lvert\chi_{\rho_j}(\phi(w)) \right\rvert+ \sum_v \bigg\lvert\sum_{u \neq v}  \rho_j(\phi(w))[u,v] \bigg\rvert\right)^2$$
	$$\geq 1-\frac{1}{d^2}(d-\tau(n)+d\delta)^2 \geq 2\left(\frac{\tau(n)}{d}-\delta\right)-\left(\frac{\tau(n)}{d}-\delta \right)^2 \geq \frac{\tau(n)}{d}-\delta,$$ 
	where the last inequality follows from the fact that $\tau(n) \leq d$. If $\mathcal{F}$ is a diagonal DFR, then $\rho_j(\phi(w))$ is a diagonal matrix, which implies $\delta=0$. In this case, if $\phi(w) \in G_j$, then $\Pr[r=0] \geq \frac{\tau(n)}{d}$.
\end{proof}

The preceding lemma allows a 2QCFA to perform the needed measurements of any diagonal DFR. We next consider the case of general DFRs.

\begin{definition2}\label{def:2qcfa:measureRound} Using the notation of Definition~\ref{def:2qcfaRound}, we define the following additional 2QCFA subroutines.
	\begin{enumerate}[(a)]
		\item \label{def:2qcfa:measureRound:reset} A \textit{reset} consists of $A$ moving its head directly to the right end of the tape, without altering its quantum register. That is to say, when reading $\#_L$ or any $\sigma \in \Sigma$, $A$ must perform the identity transformation on its quantum register and move its head one step to the right. When $\#_R$ is encountered for the first time, $A$ must again perform the identity transformation on its quantum register and $A$ must not move its head, after which the reset is complete.
		\item \label{def:2qcfa:measureRound:multi} For $p \in \mathbb{N}_{\geq 1}$, a $[\leq p]$-\text{pass measurement round} of $A$ on input $w$ consists of $A$ performing at most $p$ measurement rounds, where the overall result is the AND of the results of individual measurement rounds, and which stops as soon as any result of $0$ is obtained. Formally, we define a $[\leq p]$-pass measurement round $\mathcal{M}\left[(\pi_1,\ket{\psi_1},t_1),\ldots,(\pi_p,\ket{\psi_p},t_p)\right]$ as follows. Initialize a counter $j=1$ ($A$ keeps track of $j$ using its classical states). $A$ repeatedly does the following: $A$ performs the measurement round $\mathcal{M}(\pi_j,\ket{\psi_j},t_j)$ producing the result $r_j$, if $r_j=0$ or $j=p$, we are done and the result is $r_j$, otherwise (in particular, notice this requires $r_j=1$ and so the quantum register is $\ket{q_1}$) $A$ increments the counter to $j+1$, performs a reset, and continues (and of course does \textit{not} continue to remember $r_j$).
	\end{enumerate} 
\end{definition2}

\begin{lemma}\label{thm:algorithm:measureRep:nonDiag}
	Using the notation of \Cref{def:goodSetForDFR} and \Cref{thm:algorithm:measureRep:diag}, let $P_v \in \U(d,\overline{\mathbb{Q}})$ denote an arbitrary permutation matrix with a $1$ in entry $(1,v)$, $\forall v \in \{1,\ldots,d\}$. The result $r \in \{0,1\}$ of the $[\leq (d+1)]$-pass measurement round $\mathcal{M}\left[(\rho_j,\ket{1},F),(\rho_j,\ket{q_1},P_1),(\rho_j,\ket{q_2},P_2),\ldots,(\rho_j,\ket{q_d},P_d)\right]$ satisfies the following.
	\begin{enumerate}[(a)]
		\item \label{thm:algorithm:measureRep:nonDiag:completeness} (Perfect Completeness) If $\phi(w)=1_G$, then $\Pr[r=1]=1$.
		\item \label{thm:algorithm:measureRep:nonDiag:soundness} (Soundness) If $\phi(w) \in G_j$, then $\Pr[r=0] \geq \frac{(\tau(n))^2}{4d^3}$.
	\end{enumerate}
\end{lemma}
\begin{proof}
	If $\phi(w)=1_G$, then $\rho_j(\phi(w))=I_d$; this immediately implies all measurements performed have result $1$ with certainty, which then implies $\Pr[r=1]=1$. If $\phi(w) \in G_j$, then, by \Cref{thm:algorithm:measureRep:diag}, the result $r_1$ of the first measurement round satisfies $\Pr[r_1=0] \geq \frac{\tau(n)}{d}-\delta$. If $\delta \leq \frac{\tau(n)}{2d}$, then $\Pr[r_1=0] \geq \frac{\tau(n)}{2d}$; as $\Pr[r=0] \geq \Pr[r_1=0]$, the claim has been proven in this case. 
	
	Suppose instead that $\delta > \frac{\tau(n)}{2d}$. Fix $v'$ such that $\delta=\lvert \sum_{u \neq v'} \rho_j(\phi(w))[u,v'] \rvert$. Notice that $P_{v'} M \ket{q_{v'}}$ is of the form $\sum_h \beta_h \ket{q_h}$ where the $\beta_h$ are a permutation of the entries in column $v'$ of $M$, and $\beta_1=M_{v',v'}$. Let $p_{v+1}$ denote the probability that $B$ performs the $(v+1)^{\text{th}}$ quantum measurement (recall that a multiple pass measurement round will stop as soon as a result of $0$ is obtained) and let $r_{v+1}$ denote the result of that measurement, assuming that it is performed. Then, $$\Pr[r_{v'+1}=0] = \sum_{h > 1} \lvert\beta_h\vert^2=\sum_{u \neq v'} \lvert \rho_j(\phi(w))[u,v'] \rvert^2 \geq \frac{1}{d}\left(\sum_{u \neq v'} \lvert \rho_j(\phi(w))[u,v'] \rvert \right)^2 \geq \frac{1}{d} \delta^2 \geq \frac{(\tau(n))^2}{4d^3}.$$ Therefore,
	$$\Pr[r=0] \geq (1-p_{v'+1})+\Pr[r_{v'+1}=0]p_{v'+1}\geq \Pr[r_{v'+1}=0] \geq \frac{(\tau(n))^2}{4d^3}.$$ \qedhere	
\end{proof}

In the unbounded-error case, we have the following.

\begin{lemma}\label{thm:algorithm:unboundedMeasure}
	Suppose $\mathcal{F}=\{\rho_1,\ldots,\rho_k\}$ is an unbounded-error $[k,d,\mathbb{A}]$-DFR (or PDFR). The result $r$ of the $[\leq (d+1)]$-pass measurement round $\mathcal{M}\left[(\rho_j,\ket{1},F),(\rho_j,\ket{q_1},P_1),\ldots,(\rho_j,\ket{q_d},P_d)\right]$ satisfies the following.
	\begin{enumerate}[(a)]
		\item \label{thm:algorithm:unboundedMeasure:completeness} (Perfect Completeness) If $\phi(w)=1_G$, then $\Pr[r=1]=1$.
		\item \label{thm:algorithm:unboundedMeasure:soundness} (Soundness) If $\phi(w) \in G_j$, then $\Pr[r=0] >0$.
	\end{enumerate}
\end{lemma}
\begin{proof}
	Precisely analogous to the proof of \Cref{thm:algorithm:measureRep:nonDiag}.
\end{proof}

Finally, we consider unbounded-error MO-1QFA.

\begin{lemma}\label{thm:algorithm:oneWayUnboundedMeasure}
	Suppose $\mathcal{F}=\{\rho_1,\ldots,\rho_k\}$ is an unbounded-error $[k,d,\mathbb{A}]$-DFR (or PDFR). There is a MO-1QFA $B$ with $2 (kd)^2$ basis states such that, if $\phi(w)=1_G$, then $\Pr[B \text{ accepts } w]=1$, and if $\phi(w)\neq 1_G$, then $\Pr[B \text{ rejects } w]=>0$.
\end{lemma}
\begin{proof}
	By (the unbounded-error analogue of) \Cref{thm:dfr:conv}(\ref{thm:dfr:conv:combine}), we have an unbounded-error $[1,kd,\mathbb{A}]$-DFR (or PDFR) $\{\pi\}$. By a straightforward application of the well-known Hadamard test, we may determine if $\lvert \chi_{\pi}(w) \rvert < kd$. We omit the details.
\end{proof}

\subsection{Constructions of 2QCFA for Word Problems}\label{sec:2QCFAword:algorithm}

Now, by combining the results of the previous section, the constructions of DFRs from Section~\ref{sec:dfr:constructions}, and standard techniques from computational group theory, we show that 2QCFA can recognize the word problems of a wide class of groups. 

\begin{lemma}\label{thm:algorithm:polyDirectWordAlg}
	Consider a group $G=\langle S|R \rangle$, with $S$ finite, and let $W_G=W_{G =\langle S|R \rangle}$. Suppose $G$ has a diagonal $[k,d,C_1 n^{-C_2},\mathbb{A}]$-DFR (or PDFR), for some $C_1,C_2 \in \mathbb{R}_{>0}$. Then $\forall \epsilon \in \mathbb{R}_{>0}$, we have $W_G \in \mathsf{coR2QCFA}(n^{\lceil C_2 \rceil+2},\epsilon,d,\overline{\mathbb{Q}} \cup \mathbb{A})$.
\end{lemma}
\begin{proof}
	Define the subsets $G_j \subseteq G_{\neq 1}$ as in \Cref{def:goodSetForDFR}, and observe that $G_{\neq 1}=\cup_j G_j$. The 2QCFA $A$ will recognize $W_G$ by running the subroutine of Lemma~\ref{thm:algorithm:measureRep:diag}, for each $j$. If $\phi(w) \neq 1_G$, then, for at least some $j$, this subroutine will, with sufficient probability, produce a result that allows one to conclude with certainty, that $\phi(w) \neq 1_G$, at which point $A$ will immediately reject. To assure that $w$ for which $\phi(w)=1_G$ are accepted, $A$ will periodically run a subroutine that accepts with some small probability and continues otherwise, using the technique from Ambainis and Watrous \cite{ambainis2002two}. In particular, for $m,y \in \mathbb{N}$, let $\mathcal{R}(m,y)$ denote the subroutine that, on an input of length $n \in \mathbb{N}$ produces a result $b \in \{0,1\}$, where $\Pr[b=1]=(n+1)^{-m} 2^{-y}$, within expected running time $O(n^2)$ (see \cite{ambainis2002two} for details; in brief, if the 2QCFA starts with its head over the first symbol to the right of $\#_L$ and performs an unbiased one-dimensional random walk along the tape until either of the end-markers are encountered, then the probability that $\#_R$ is the first end-marker encountered is $(n+1)^{-1}$; by repeating this procedure $m$ times, and generating unbiased random bits $y$ times, the desired $b$ can be produced).
	
	We now fill in the details. $A$ has the quantum basis states $\ket{q_1},\ldots,\ket{q_d}$, where $q_1$ is the quantum start state. $A$ performs the following procedure.
	\vspace{5mm}
	\newline Use the classical states to store a counter $j \in \{1,\ldots,k\}$, initialized to $1$
	\newline Repeat indefinitely:
	\newline \hspace*{8mm} Move the head to the right end of the tape, leaving the quantum register unchanged
	\newline \hspace*{8mm} Run the subroutine of Lemma~\ref{thm:algorithm:measureRep:diag} with $\rho_j$ producing the result $r$
	\newline \hspace*{8mm} If $r=0$ then \underline{reject}
	\newline \hspace*{8mm} Add $1$ to $j$, where the addition is performed modulo $k$ 
	\newline \hspace*{8mm} If $j=k$ then
	\newline \hspace*{16mm} Run the subroutine $\mathcal{R}(\lceil C_2 \rceil,\left\lceil \log(\frac{\epsilon C_1}{d})\right\rceil) $, giving the result $b$
	\newline \hspace*{16mm} If $b=1$ then \underline{accept}
	\vspace{5mm}

	We now show that $A$ has the claimed parameters.  Clearly, $A$ has $d$ basis states and the transition amplitudes of $A$ belong to $\overline{\mathbb{Q}} \cup \mathbb{A}$. To see the remaining claims, fix a string $w$ and let $n$ denote its (string) length. Consider a subcomputation of the above computation of $A$ that begins when the counter $j=1$ and $A$ is at the beginning of the ``Repeat indefinitely" loop, and ends as soon as $A$ accepts or rejects, or after $k$ complete iterations of the ``Repeat indefinitely" loop. Let $p_{acc}$ and $p_{rej}$ denote, respectively, the probability that such a subcomputation ends with $A$ accepting or rejecting. Let $E_j$ denote the event that such a subcomputation actually runs the subroutine of Lemma~\ref{thm:algorithm:measureRep:diag} with $\rho_{j}$ (note that the only way this does not happen is if $A$ has already rejected for some $\widetilde{j}<j$), let $p_j$ denote the probability that $E_j$ occurs, and let $r_j$ denote the result produced by this subroutine, if $E_j$ occurs. Notice that $$\Pr[b=1|E_k]=2^{-\left\lceil \log(\frac{\epsilon C_1}{d})\right\rceil} (n+1)^{-\lceil C_2 \rceil} >0.$$
	
	First, suppose $w \not \in W_G$. There is at least one $j'$ such that $\phi(w) \in Y_{j'}$. Therefore, when the counter $j=j'$, \Cref{thm:algorithm:measureRep:diag}(\ref{thm:algorithm:measureRep:diag:soundness}) guarantees that $\Pr[r_{j'}=0|E_{j'}] \geq \frac{C_1}{d} n^{-C_2}$. Notice that the event that $A$ rejects in such a subcomputation is the (disjoint) union of the event $A$ rejects before step $j'$ (i.e., $E_{j'}$ does not occur) and the event $A$ rejects at step $j'$ or later. Therefore, $$p_{rej}=(1-p_{j'})1+\sum_{j \geq j'} p_j \Pr[r_j=0|E_j]\geq (1-p_{j'})+ p_{j'} \Pr[r_{j'}=0|E_{j'}]\geq \Pr[r_{j'}=0|E_{j'}] \geq \frac{C_1}{d} n^{-C_2}.$$ We also have $$p_{acc}=p_k \Pr[b=1|E_k] <\Pr[b=1|E_k]=2^{-\left\lceil \log(\frac{\epsilon C_1}{d})\right\rceil} (n+1)^{-\lceil C_2 \rceil} \leq \epsilon \frac{C_1}{d} (n+1)^{-\lceil C_2 \rceil} \leq \epsilon p_{rej}.$$ As we repeat such subcomputations until $A$ either accepts or rejects, we have $$\Pr[A \text{ rejects } w|w \not \in W_G]=\frac{p_{rej}}{p_{acc}+p_{rej}} \geq \frac{p_{rej}}{\epsilon p_{rej}+p_{rej}}=\frac{1}{1+\epsilon} \geq 1-\epsilon.$$
	
	Next, instead suppose $w \in W_G$. Then \Cref{thm:algorithm:measureRep:diag}(\ref{thm:algorithm:measureRep:diag:completeness}) guarantees that every use of the subroutine of Lemma~\ref{thm:algorithm:measureRep:diag} will produce $r=1$. This implies $p_{rej}=0$, $p_k=1$, and $$p_{acc}=p_k \Pr[b=1|E_k]=2^{-\left\lceil \log(\frac{\epsilon C_1}{d})\right\rceil} (n+1)^{-\lceil C_2 \rceil} \geq \epsilon \frac{C_1}{2d} (n+1)^{-\lceil C_2 \rceil}>0.$$ As we repeat such subcomputations until $A$ either accepts or rejects, we have $$\Pr[A \text{ accepts } w|w \in W_G]=\frac{p_{acc}}{p_{acc}+p_{rej}}=1.$$ 
	
	This completes the proof of the claim that $A$ recognizes $W_G$ with one-sided error $\epsilon$. Lastly, to see that $A$ has the claimed expected running time, let $p_{halt}$ denote the probability that any given subcomputation of the above form ends with $A$ halting (i.e., accepting or rejecting). When $w \in W_G$, $$p_{halt}=p_{acc}+p_{rej}\geq\epsilon \frac{C_1}{2d} (n+1)^{-\lceil C_2 \rceil}.$$ When $w \not \in W_G$, $$p_{halt}=p_{acc}+p_{rej} \geq p_{rej} \geq \frac{C_1}{2d} n^{-C_2} \geq \epsilon \frac{C_1}{2d} (n+1)^{-\lceil C_2 \rceil}.$$ Therefore the expected number of executions of such subcomputations is $O(n^{\lceil C_2 \rceil})$. Each subcomputation of the above form consists of at most $k$ passes through the ``Repeat indefinitely" loop. Each pass involves a single use of the subroutine of Lemma~\ref{thm:algorithm:measureRep:diag}, which runs in time $O(n)$; additionally, the pass in which the counter $j=k$ also involves a single use of the subroutine $\mathcal{R}$, which runs in time $O(n^2)$. Therefore, $A$ runs in expected time $O(n^{\lceil C_2 \rceil+2})$, as desired. 		
\end{proof}

\begin{lemma}\label{thm:algorithm:genDirectWordAlg}
	Consider a group $G=\langle S|R \rangle$, with $S$ finite, and let $W_G=W_{G =\langle S|R \rangle}$. If $G$ has a $[k,d,C_1^{-n},\mathbb{A}]$-DFR (or PDFR), for some $C_1 \in \mathbb{R}_{\geq 1}$, then $\forall \epsilon \in \mathbb{R}_{>0}$, $\exists C_2 \in \mathbb{R}_{\geq 1}$ such that $W_G \in \mathsf{coR2QCFA}(C_2^n,\epsilon,d,\overline{\mathbb{Q}} \cup \mathbb{A})$.
\end{lemma}
\begin{proof}
	We proceed almost exactly as in the proof of \Cref{thm:algorithm:polyDirectWordAlg}, with the only modification arising from the fact that the substantially weaker bound on the parameter $\tau$ of the DFR has a corresponding decrease in the probability that the subroutine of Lemma~\ref{thm:algorithm:measureRep:nonDiag} can distinguish $w$ with $\lvert \chi_{\rho_j}(\phi(w))\rvert=d$ from $w$ with $\lvert \chi_{\rho_j}(\phi(w))\rvert \neq d$. As before, $A$ will periodically run a subroutine that accepts with some small probability, though the above issue requires that this is done with a substantially smaller probability than in the proof of \Cref{thm:algorithm:polyDirectWordAlg}. 
	
	$A$ has the quantum basis states $\ket{q_1},\ldots,\ket{q_d}$, where $q_1$ is the quantum start state. For $p \in \overline{\mathbb{Q}} \cap [0,1]$, let $\mathcal{B}(p)$ denote the subroutine that produces a biased random Boolean value $x$, such that $\Pr[x=1]=p$, which operates as follows. We start with the quantum register in the superposition $\ket{q_1}$. Let $\ket{\psi}=\sqrt{p} \ket{q_1}+\sqrt{1-p} \ket{q_2}$. We then perform the unitary transformation $T_{\ket{q_1} \rightarrow \ket{\psi}}$, followed by the quantum measurement with respect to the partition $B_0=\{2\ldots,d\},B_1=\{1\}$. The result $1$ occurs with probability $p$. If the result is $0$, we then perform the unitary transformation $T_{\ket{q_2} \rightarrow \ket{q_1}}$ to return the quantum register to the superposition $\ket{q_1}$. The head of the 2QCFA does not move during this subroutine.
	
	For $p \in \overline{\mathbb{Q}} \cap [0,1]$, $y \in \mathbb{N}$, let $\mathcal{R}'(p,y)$ denote the subroutine that, on an input of length $n \in \mathbb{N}$ produces a result $b \in \{0,1\}$, where $\Pr[b=1]=p^n 2^{-y}$, and has running time $O(n)$. $\mathcal{R}'(p,y)$ operates by scanning the tape once, from left to right. On symbols other than the end-markers, $\mathcal{B}(p)$ is run; if the result is $0$, the subroutine immediately halts with the result of $0$, otherwise it continues reading the next symbol. When the right end-marker $\#_R$ is encountered, the subroutine generates up to $y$ unbiased bits, one after the other. If any of these bits are $0$, the subroutine immediately halts with the result of $0$; if all $y$ bits are $1$, the subroutine halts with the result of $1$. Notice that the transition amplitudes needed to implement $\mathcal{R}'$ are all algebraic numbers.
	
	$A$ performs the following procedure.
	\vspace{5mm}
	\newline Use the classical states to store a counter $j \in \{1,\ldots,k\}$, initialized to $1$
	\newline Repeat indefinitely:
	\newline \hspace*{8mm} Move the head to the right end of the tape, leaving quantum register unchanged
	\newline \hspace*{8mm} Run the subroutine of Lemma~\ref{thm:algorithm:measureRep:nonDiag} with $\rho_j$ producing the result $r$
	\newline \hspace*{8mm} If $r=0$ then \underline{reject}
	\newline \hspace*{8mm} Add $1$ to $j$, where the addition is performed modulo $k$ 
	\newline \hspace*{8mm} If $j=k$ then
	\newline \hspace*{16mm} Run the subroutine $\mathcal{R}'(\frac{1}{\lceil C^2 \rceil},\lceil \log(\frac{\epsilon}{4d^4})\rceil)$, giving the result $b$
	\newline \hspace*{16mm} If $b=1$ then \underline{accept}
	\vspace{5mm}
	
	All remaining parts of the proof are identical to that of \Cref{thm:algorithm:polyDirectWordAlg}, and so we omit the details.		
\end{proof}

\begin{lemma}\label{thm:algorithm:unboundedWordAlg}
	Consider a group $G=\langle S|R \rangle$, with $S$ finite, and let $W_G=W_{G =\langle S|R \rangle}$. If $G$ has an unbounded-error $[k,d,\mathbb{A}]$-DFR (or PDFR), then $W_G \in \mathsf{coN2QCFA}(n,d,\overline{\mathbb{Q}} \cup \mathbb{A})$.
\end{lemma}
\begin{proof}
	The 2QCFA $A$ operates by using \Cref{thm:algorithm:unboundedMeasure} to check if $\lvert \chi_{\rho_j}(\phi(w))\rvert \neq d$, for each $j$. If this subroutine produces the result $0$ for some $j$, then $A$ rejects; otherwise, $A$ accepts. It is immediate that $A$ recognizes $W_G$ with negative one-sided bounded error, and that $A$ has the claimed parameters. 
\end{proof}

We now show that, if $H$ is a finite-index subgroup of $G$, a 2QCFA that recognizes $W_G$ can be constructed from a 2QCFA that recognizes $W_H$.

\begin{lemma}\label{thm:algorithm:overgroup}
	Consider a group $H=\langle S_H|R_H\rangle$, with $S_H$ finite, and suppose that $A_H$ is a 2QCFA that recognizes $W_H$, which operates in the manner of our proofs of \Cref{thm:algorithm:polyDirectWordAlg,thm:algorithm:genDirectWordAlg,thm:algorithm:unboundedWordAlg}. Further suppose $G$ is a group such that $H \leq G$ and $[G:H]$ is finite. Then $G$ admits a presentation $G=\langle S_G|R_G \rangle$, with $S_G$ finite, such that there is a 2QCFA $A_G$ that recognizes $W_G$. Moreover, $A_G$ has the same acceptance criteria, asymptotic expected running time, number of quantum basis states, and class of transition amplitudes as $A_H$.
\end{lemma}
\begin{proof}
	Following (essentially) \cite{muller1983groups} (with the exception that we do not assume $H$ is a \textit{normal} subgroup of $G$), we now construct a convenient presentation for $G$. We begin by establishing some notation. Let $l=[G:H]$, and let $g_1,\ldots,g_l$ denote a complete family of left coset representatives of $H$ in $G$, where $g_1=1_G$. We assume for notational convenience that $S_H \cap S_H^{-1} = \emptyset$ (and so, in particular, $1_H \not \in S_H$). Let $\Sigma_H=S_H \sqcup S_H^{-1}$, $S_G=S_H \sqcup (g_2,\ldots,g_l)$, and $\Sigma_G=S_G \cup S_G^{-1}$. Let $\phi_H:\Sigma_H^* \rightarrow H$ and $\phi_G:\Sigma_G^* \rightarrow G$ be the natural maps. Let $T_l=\{1,\ldots,l\}$. 
	
	As the $g_i$ are a complete family of left coset representatives of $H$ in $G$, every element $g \in G$ can be expressed uniquely as some $g_i h$, where $i \in T_l$ and $h \in H$. In particular, for any $\sigma \in \Sigma_G$ and $j \in T_l$, consider the element $\sigma g_j \in G$; there is unique $i \in T_l$ and $h \in H$ such that $\sigma g_j =g_i h$. Therefore, we can define functions $\alpha: \Sigma_G \times T_l \rightarrow T_l$ and $\beta:\Sigma_G \times T_l \rightarrow H$, such that $$\sigma g_j=g_{\alpha(\sigma,j)} \beta(\sigma,j), \ \forall \sigma \in \Sigma_G,\forall j \in T_l.$$ Let $\tau:H \rightarrow F(S_H)$ be the function that takes each $h \in H$ to some element in the free group on $S_H$ such that $h=\tau(h)$, as elements of $H$. Then $G$ has presentation $\langle S_G|R_G \rangle$, where $S_G$ is as defined above and $$R_G=R_H \cup \left\{g_{\alpha(\sigma,j)} \tau(\beta(\sigma,j)) g_j^{-1} \sigma^{-1} :\sigma \in \Sigma_G,j \in T_l\right\}.$$ 
	
	We now construct a 2QCFA $A_G$ that recognizes $W_G:=W_{G=\langle S_G|R_G \rangle}$. Consider an input $w \in \Sigma_G^*$. For any $p \in \{0,\ldots,\lvert w \rvert\}$, let $w^p=w_{\lvert w \rvert-p+1} \cdots w_{\lvert w \rvert}$ denote the suffix of $w$ of length $p$; in particular, $w^0$ is the empty string. $A_G$ must determine if $\phi_G(w)=1_G=g_1 1_H$. The key idea is that $A_G$ will make many right-to-left passes over its input, such that, after $A_G$ has read the suffix $w^p$, if $\phi_G(w^p)=g_m h$, then $A_G$ will have the values $m \in T_l$ and $h \in H$ ``stored" in its internal state, in an appropriate sense. Namely, $A_G$ will keep track of $m \in T_l$ using its classical states, and $A_G$ will keep track of $h$ by simulating $A_H$. 
	
	We now fill in the details. $A_G$ has the same quantum basis states as $A_H$, which we will denote $\ket{q_1},\ldots,\ket{q_d}$, and quantum start state $q_1$. $A_G$ begins by moving its head to the far right end of the tape, leaving its quantum register in the superposition $\ket{q_1}$. $A_G$ will store a value $t \in T_l$ using its classical states, where $t$ is initialized to $1$. $A_G$ then repeatedly scans its input in the manner prescribed by $A_H$, i.e., $A_G$ makes many right-to-left passes reading the input word $w$, and $A_G$ also performs the simulated coin flipping via random walks of $A_H$. During each right-to-left pass, $A_G$ will maintain the property that after reading the suffix $w^p$, if $\phi_G(w^p)=g_m h$, then the stored value $t=m$ and $A_N$ will have been simulated on a string $\widehat{w^p} \in \Sigma_H^*$ (read ``backwards"), where $\phi_H(\widehat{w^p})=h$. 
	
	$A_G$ accomplishes this as follows. Suppose $A_G$ has already read the particular suffix $w_p$ and $\phi_G(w^p)=g_m h$, and is now about to read the next symbol, $\sigma:=w_{\lvert w \rvert - p}$. After reading $\sigma$, we want $A_G$ to update its internal state (both classical and quantum) to correspond to the word $w^{p+1}=\sigma \circ w^p$. By construction, $\sigma g_m=g_{\alpha(\sigma,m)} \beta(\sigma,m)$, and so $$\phi_G(w^{p+1})=\phi_G(\sigma \circ w^p)=\phi_G(\sigma) \phi_G(w^p)=\sigma g_m h=g_{\alpha(\sigma,m)} \beta(\sigma,m) h.$$ Define the function $\widehat{\beta}:\Sigma_G \times T_l \rightarrow \Sigma_H^*$ such that $\widehat{\beta}(\kappa,j)$ is any word in $\Sigma_H^*$ of minimum (string) length such that $\phi_N(\widehat{\beta}(\kappa,j))=\beta(\kappa,j)$, $\forall \kappa \in \Sigma_G, \forall j \in T_l$. $A_G$ then updates its stored value $t \in T_l$ from $m$ to $\alpha(\sigma,m)$ and simulates $A_H$ on $\widehat{\beta}(\sigma,m)$. That is to say, at this point $A_H$ has been simulated on the string $\widehat{w^p}$, where $\phi_H(\widehat{w^p})=h$; $A_G$ then feeds the string $\widehat{\beta}(\sigma,m)$ to $A_H$ (from right-to-left), after which $A_H$ will have been simulated on $\widehat{\beta}(\sigma,m) \circ \widehat{w^p}$, as desired. During this process of feeding the string $\widehat{\beta}(\sigma,m)$ to $A_H$, $A_G$ does not move its head.
	
	All that remains is to define the acceptance criteria of $A_G$. Suppose $A_G$ has just made a complete pass over the input, simulating $A_H$ along the way, and then possibly also performed a simulated coin-flipping procedure, if $A_H$ so demanded. $A_G$ also has the value $m$ in its internal state, such that $\phi_G(w)=g_m h$. At this point (the simulation of) $A_H$ may or may not have halted. $A_G$ behaves as follows. If $m\neq 1$, $A_G$ immediately rejects. If $m=1$, then if $A_H$ has halted (accepting or rejecting the input), then $A_G$ halts, accepting if $A_H$ accepted and rejecting if $A_H$ rejected. If $m=1$ and $A_H$ has not halted, $A_G$ continues. It immediately follows from the above argument that $A_G$ recognizes $W_G$ and that $A_G$ has all the claimed properties.  	
\end{proof}

Using the above results, and the constructions of DFR from Section~\ref{sec:dfr}, the main theorems stated in the introduction straightforwardly follow.

\begin{proof}[Proof of Theorem~\ref{thm:main:abelian}]
	Fix $G \in \widehat{\Pi}_1$. By \Cref{thm:dfr:virtAbelian}(\ref{thm:dfr:virtAbelian:alg}), $G$ virtually has a diagonal algebraic $[K_1,2,D_2 n^{-D_1}]$-DFR, for some $K_1 \in \mathbb{N}_{\geq 1}$ and $D_1,D_2 \in \mathbb{R}_{>0}$ (where $D_1$ is a universal constant that does not depend on $G$). By \Cref{thm:algorithm:polyDirectWordAlg,thm:algorithm:overgroup}, we conclude $W_G \in \mathsf{coR2QCFA}(n^{\lceil D_1 \rceil+2},\epsilon,2,\overline{\mathbb{Q}})$. Similarly, by \Cref{thm:dfr:virtAbelian}(\ref{thm:dfr:virtAbelian:nonAlg}), with $\delta=0.9$, $G$ virtually has a diagonal $[K_2,2,D_3 n^{-0.9},\widetilde{\mathbb{C}}]$-DFR. By \Cref{thm:algorithm:polyDirectWordAlg,thm:algorithm:overgroup}, $W_G \in \mathsf{coR2QCFA}(n^3,\epsilon,2,\widetilde{\mathbb{C}})$.
\end{proof}

\begin{proof}[Proof of Theorem~\ref{thm:main:freeDirProd}]
	Follows from \Cref{thm:dfr:virtFree}, \Cref{thm:algorithm:genDirectWordAlg}, and \Cref{thm:algorithm:overgroup}.
\end{proof}

\begin{proof}[Proof of Theorem~\ref{thm:main:projEmbed}]
	Follows from \Cref{thm:pdfr:generic}, \Cref{thm:algorithm:genDirectWordAlg}, and \Cref{thm:algorithm:overgroup}.
\end{proof}

\begin{proof}[Proof of Theorem~\ref{thm:main:pi3}]
	Follows from \Cref{thm:unboundedDFR:pi3} and \Cref{thm:algorithm:unboundedWordAlg}.
\end{proof}

\begin{proof}[Proof of Theorem~\ref{thm:main:totallyGeneric}]
	By the assumption of the theorem, $G$ has a (finitely generated) finite index subgroup that has an unbounded-error $[k,d,\mathbb{C}]$-PDFR. For 2QCFA, the claim follows from \Cref{thm:algorithm:unboundedWordAlg} and \Cref{thm:algorithm:overgroup}; for MO-1QFA, the claim follows from \Cref{thm:algorithm:oneWayUnboundedMeasure}.
\end{proof}

\section{Discussion}\label{sec:discussion}

\subsection{Computational Complexity of the Word Problem}\label{sec:discussion:wordComplex}

We now compare the results that we have obtained concerning the ability of a 2QCFA to recognize certain group word problems with existing results  for ``simple" classical and quantum models. We use the following notation for complexity classes: $\mathsf{REG}$ denotes the regular languages (languages recognized by deterministic finite automata), $\mathsf{CFL}$ (resp. $\mathsf{DCFL}$) denotes the context-free (resp. deterministic context-free) languages (languages recognized by nondeterministic (resp. deterministic) pushdown automata), $\mathsf{OCL}$ (resp. $\mathsf{DOCL}$) denotes the one-counter (resp. deterministic one-counter) languages (languages recognized by nondeterministic (resp. deterministic) pushdown automata where the stack alphabet is limited to a single symbol), $\mathsf{poly{-}CFL}$ (resp. $\mathsf{poly{-}DCFL}$, $\mathsf{poly{-}OCL}$, $\mathsf{poly{-}DOCL}$) denotes the intersection of finitely many context-free (resp. deterministic context-free, one-counter, deterministic one-counter) languages, and $\mathsf{L}$ denotes deterministic logspace (languages recognized by deterministic Turing machines with read-only input tape and read/write work tape of size logarithmic in the input). 

Using the notation of \Cref{sec:intro:results}, we write $\widehat{\Pi}_0$ (resp. $\widehat{\Pi}_1$, $\widehat{\Sigma}_1$, $\widehat{\Pi}_2$) for the finitely-generated groups that are virtually cyclic (resp. abelian, free, a subgroup of a direct product of finitely many finite-rank free groups). We also write $\widehat{\{1\}}$ for the finite groups (i.e., the virtually trivial groups), and $\mathcal{L}$ for the set of all finitely generated groups $G$ that are linear groups over some field of characteristic $0$. The following proposition, which collects the results of many authors, demonstrates the extremely strong relationship between the computational complexity of $W_G$ and certain algebraic properties of $G$.

\begin{proposition}
	(\cite{anisimov1971group, herbst1991subclass, holt2008groups, brough2014groups, muller1983groups, dunwoody1985accessibility, anisimov1972some, muller1985theory, lipton1977word}) Let $G$ be a finitely generated group with word problem $W_G$. The following statements hold.
	\begin{enumerate}[(i)]
		\item $G \in \widehat{\{1\}} \Leftrightarrow W_G \in \mathsf{REG}.$
		\item $G \in \widehat{\Pi}_0 \Leftrightarrow W_G \in \mathsf{OCL} \Leftrightarrow W_G \in \mathsf{DOCL}.$
		\item $G \in \widehat{\Pi}_1 \Leftrightarrow W_G \in \mathsf{poly{-}OCL}\Leftrightarrow W_G \in \mathsf{poly{-}DOCL}.$
		\item $G \in \widehat{\Sigma}_1 \Leftrightarrow W_G \in \mathsf{CFL} \Leftrightarrow W_G \in \mathsf{DCFL}.$
		\item $G \in \widehat{\Pi}_2 \Rightarrow W_G \in \mathsf{poly{-}DCFL} \subsetneq \mathsf{poly{-}CFL}.$
		\item $G \in \mathcal{L} \Rightarrow W_G \in \mathsf{L}.$
	\end{enumerate}
\end{proposition} 

\begin{proof}
	Statements $(i),(ii),(iii),(v),\text{ and } (vi)$ were shown, respectively, in \cite{anisimov1971group},\cite{herbst1991subclass},\cite{holt2008groups}, \cite{brough2014groups}, and \cite{lipton1977word}. In \cite{muller1983groups}, it was shown that $G$ is free if and only if $W_G \in \mathsf{CFL}$ and $G$ is accessible, in \cite{dunwoody1985accessibility}, it was shown that all finitely presented groups are accessible, and in \cite{anisimov1972some} it was shown that all context-free groups are finitely presented, which implies the first equivalence in $(iv)$. The second equivalence in $(iv)$ was shown in \cite{muller1985theory}. 
\end{proof}

It is particularly interesting that, while there are strict inclusions $\mathsf{DCFL} \subsetneq \mathsf{CFL}$, $\mathsf{DOCL} \subsetneq \mathsf{OCL}$, and $\mathsf{poly{-}DOCL} \subsetneq \mathsf{poly{-}OCL}$, there are no groups whose word problem witnesses any of these separations. That is to say, the deterministic and non-deterministic versions of each of these models can recognize word problems for precisely the same class of groups. 

Our results have a close correspondence to the above mentioned results. By \Cref{thm:main:abelian} (resp. \Cref{thm:main:freeDirProd}), $\forall G \in \widehat{\Pi}_1 \supsetneq \widehat{\Pi}_0 \supsetneq \widehat{\{1\}}$ (resp. $\forall G \in \widehat{\Pi}_2 \supsetneq\widehat{\Pi}_1 \cup \widehat{\Sigma}_1 \supsetneq \widehat{\Pi}_0 \supsetneq \widehat{\{1\}}$), $W_G$ is recognized with one-sided bounded error, in expected polynomial (resp. exponential) time, by a 2QCFA with a single qubit and algebraic number transition amplitudes. Moreover, if allowed a quantum register of any constant size, such a 2QCFA may recognize the word problem of any group $G \in \mathcal{Q}$, where $\mathcal{Q}$ denotes the class of groups for which \Cref{thm:main:projEmbed} applies, with one-sided bounded error in expected exponential time. Of course, as our fundamental approach to solving the group word problem is to construct a DFR for a group $G$, and as any such DFR yields a faithful finite-dimensional unitary representation of $G$, any such $G \in \mathcal{L}$.

In a companion paper \cite{remscrim2019lower}, we establish a lower bound on the running time of any 2QCFA (with any size quantum register and no restrictions placed on its transition amplitudes) that recognizes a word problem $W_G$ with bounded error (even under the more generous notion of two-sided bounded error); more strongly, we establish a lower bound on the running time of \textit{any quantum Turing machine} that uses \textit{sublogarithmic} space, though we will not discuss that here. In particular, we show that, $\forall G \in \mathcal{Q}\setminus\widehat{\Pi}_1$, $W_G$ \textit{cannot} be recognized by such a 2QCFA is expected time $2^{o(n)}$. Therefore, the algorithm exhibited in this paper for recognizing the word problem of any group $G \in \mathcal{Q}\setminus\widehat{\Pi}_1$ has (essentially) optimal expected running time; moreover, we have obtained the first provable separation between the classes of languages recognizable with bounded error by 2QCFA in expected exponential time and in expected subexponential time. In that same paper, we also show that if a 2QCFA of this most general type recognizes a word problem $W_G$ in expected polynomial time, then $G \in \mathcal{G}_{vNilp}$, where $\mathcal{G}_{vNilp}$ denotes the finitely generated virtually nilpotent groups, and $\widehat{\Pi}_1 \subsetneq \mathcal{G}_{vNilp}$. This naturally raises the following question.

\begin{openProblem}
	Is there a group $G \in (\mathcal{G}_{vNilp} \setminus \widehat{\Pi}_1)$ such that $W_G$ can be recognized by a 2QCFA with bounded error in expected polynomial time? 
\end{openProblem}

We have shown that the (three-dimensional discrete) Heisenberg group $H\in (\mathcal{G}_{vNilp} \setminus \widehat{\Pi}_1)$ is ``complete'' for this question, in the sense that if $W_H$ \textit{cannot} be recognized with bounded error by a 2QCFA in expected polynomial time, then no such $G$ can \cite{remscrim2019lower}. 

Let $\mathcal{G}_{vSolvLin}$ denote the finitely generated virtually solvable linear groups over a field of characteristic zero, and note that $\mathcal{G}_{vNilp} \subsetneq \mathcal{G}_{vSolvLin}$. Furthermore, note that $W_G \in \mathsf{L}$, $\forall G \in \mathcal{G}_{vSolvLin}$ \cite{lipton1977word}. However, every $G \in \mathcal{G}_{vSolvLin} \setminus \widehat{\Pi}_1$ \textit{does not} have a faithful finite-dimensional unitary representation (see, for instance, \cite[Proposition 2.2]{thom2013convergent}) and, therefore, does not have a DFR (even an unbounded-error DFR); this prevents the techniques of this paper from producing a 2QCFA that recognizes the corresponding $W_G$.

\begin{openProblem}
	Is there a finitely generated group $G$ that does not have a faithful finite-dimensional unitary representation (for example, any $G \in \mathcal{G}_{vSolvLin} \setminus \widehat{\Pi}_1$ or any finitely generated infinite Kazhdan group) such that $W_G$ can be recognized with bounded error by a 2QCFA at all (i.e., in any time bound)? 
\end{openProblem}

Consider the group $\mathbb{Z} * \mathbb{Z}^2 \in \Sigma_2 \subsetneq \widehat{\Pi}_3$, and note that $\mathbb{Z} * \mathbb{Z}^2  \not \in \widehat{\Pi}_2$. The complexity of $W_{\mathbb{Z} * \mathbb{Z}^2}$ has been considered by many authors and it is conjectured that $W_{\mathbb{Z} * \mathbb{Z}^2} \not \in \mathsf{poly{-}CFL}$ \cite{brough2014groups}(cf. \cite{ceccherini2015multipass}) and that $W_{\mathbb{Z} * \mathbb{Z}^2} \not \in \mathsf{coCFL}$ \cite{holt2005groups}. By \Cref{thm:main:pi3}, $W_{\mathbb{Z} * \mathbb{Z}^2}$ is recognizable with one-sided \textit{unbounded} error by a 2QCFA. We ask the following questions.

\begin{openProblem}
	Can $W_{\mathbb{Z} * \mathbb{Z}^2}$ be recognized by a 2QCFA with bounded error? More generally, is $W_{\mathbb{Z} * \mathbb{Z}^r}$ recognizable by a 2QCFA with bounded error, $\forall r \in \mathbb{N}$?
\end{openProblem}

\begin{openProblem}
	Does $\mathbb{Z} * \mathbb{Z}^2$ have an algebraic DFR. More generally, does $\mathbb{Z} * \mathbb{Z}^r$ have an algebraic DFR, $\forall r \in \mathbb{N}$? Even more generally, is the class of groups which have algebraic DFRs closed under free product?
\end{openProblem}

\begin{remark} 
	Of course, such a DFR would immediately yield a 2QCFA of the desired type for the corresponding word problem. Moreover, recall that $\Sigma_2$ consists of all groups of the form $\mathbb{Z}^{r_1} * \cdots * \mathbb{Z}^{r_m}$, for some $m,r_1,\ldots,r_m \in \mathbb{N}$, and that any such groups embeds in $\mathbb{Z} * \mathbb{Z}^r$, where $r= \max_j r_j$. By \Cref{thm:dfr:subOver}(\ref{thm:dfr:subOver:sub}), if $\mathbb{Z} * \mathbb{Z}^r$ has a DFR then $\mathbb{Z}^{r_1} * \cdots * \mathbb{Z}^{r_m}$ has a DFR with essentially the same parameters. Therefore, if all such $\mathbb{Z} * \mathbb{Z}^r$ have DFRs of the desired type, then so do all groups in $\Sigma_2$, which would then imply all groups in $\widehat{\Pi}_3$ virtually have such a DFR, by an application of \Cref{thm:dfr:prod} and \Cref{thm:dfr:subOver}(\ref{thm:dfr:subOver:sub}).
\end{remark}

We next consider known results concerning those group word problems recognizable by particular QFA variants. Ambainis and Watrous, in the paper in which the 2QCFA model was first defined \cite{ambainis2002two}, considered the languages $L_{eq}=\{a^m b^m : m \in \mathbb{N}\}$ and $L_{pal}=\{w \in \{a,b\}^*:w \text{ is a palindrome}\}$. They showed that a 2QCFA, with only two quantum basis states (i.e., a single-qubit quantum register), can recognize $L_{eq}$ (resp. $L_{pal}$) with one-sided bounded error in expected polynomial (resp. exponential) time. As noted in the introduction, while neither $L_{eq}$ nor $L_{pal}$ are group word problems, they are closely related to word problems. In particular, $L_{eq}=(a^* b^*) \cap W_{\mathbb{Z}}$. Moreover, for $w=w_1 \cdots w_n \in \{a,b\}^*$, where each $w_i \in \{a,b\}$, let $\overline{w}=w_1^{-1}\cdots w_n^{-1}\in \{a^{-1},b^{-1}\}^*$; then, for any $w \in \{a,b\}^*$, $w \in L_{pal} \Leftrightarrow w \overline{w} \in W_{F_2}$. This observation allows us to reinterpret the above results of Ambainis and Watrous in terms of group word problems. 

In addition to results of the above form, which, implicitly, study the quantum computational complexity of the word problem for certain groups, some authors have explicitly considered this question. In the following we write MO-1QFA for the measure-once one-way QFA (defined in \cite{moore2000quantum}), MM-1QFA for the measure-many one-way QFA (defined in \cite{kondacs1997power}) and 1QFA$\circlearrowleft$ for the one-way QFA with restart (defined in \cite{yakaryilmaz2010succinctness}). Let $\mathsf{S}_{\mathbb{Q}}^{=}$ denote the class of languages $L$ for which there is a probabilistic finite automaton $P$, all of whose transition amplitudes are rational numbers, such that, $\forall w \in L$, the probability that $P$ accepts $w$ is exactly $\frac{1}{2}$, and, $\forall w \not \in L$, the probability that $P$ accepts $w$ differs from $\frac{1}{2}$.

The languages $W_{F_k}$, $k \in \mathbb{N}$ can be recognized, with negative one-sided \textit{unbounded} error, by a MO-1QFA \cite{brodsky2002characterizations}. Yakaryilmaz and Say \cite{yakaryilmaz2010succinctness} showed that any language $L \in \mathsf{S}_{\mathbb{Q}}^{=}$ can be recognized by a MM-1QFA, with negative one-sided \textit{unbounded} error, and by a 1QFA$\circlearrowleft$ or 2QCFA, with negative one-sided \textit{bounded} error, in expected \textit{exponential time}. As $L_{eq}, L_{pal} \in \mathsf{S}_{\mathbb{Q}}^{=}$, this result, partially, subsumes the original result of Ambainis and Watrous \cite{ambainis2002two}. However, in addition to the (exponential) difference in expected running time in the case of $L_{eq}$, we also note that there is a significant difference between the sizes of the quantum registers of the machines produced in these two results. In particular, the 1QFA$\circlearrowleft$ and 2QCFA constructed by Yakaryilmaz and Say that recognize $L_{pal}$ have $15$ quantum basis states, as opposed to the $2$ quantum basis states of the 2QCFA constructed by Ambainis and Watrous. Similarly, as $W_{F_k} \in \mathsf{S}_{\mathbb{Q}}^{=}$, $\forall k \in \mathbb{N}$, the result of Yakaryilmaz and Say shows that the word problems of these groups can be recognized by a 2QCFA of our type; however, a direct application of their construction would yield a 2QCFA with larger quantum part than that of our construction, or that of Ambainis and Watrous. Of course, our results also apply to the 1QFA$\circlearrowleft$ model (with exponential expected running time).

\subsection{Information Compression}\label{sec:discussion:compress}

The 2QCFA constructed by Ambainis and Watrous \cite{ambainis2002two} that recognize $L_{eq}$ and $L_{pal}$ do so using only a single qubit; as they noted, this demonstrates that quantum computational models can perform a particularly interesting sort of extreme information compression. We next observe that the same phenomenon occurs in our constructions of 2QCFA. Consider a group $G=\langle S|R \rangle$, with $S$ finite, and let $W_G=W_{G=\langle S|R \rangle}$. Let $B_{G,S}(n)=\{g \in G:l_S(g) \leq n\}$ denote those elements of $G$ of length at most $n$, and let $f_{G,S}(n)=\lvert B_{G,S}(n) \rvert$ denote the \textit{growth rate} of $G$. For the remainder of this section, we ignore the uninteresting case in which $G$ is a finite group (as then $W_G \in \mathsf{REG}$), and consider only finitely generated infinite groups, where $f_{G,S}$ is necessarily a growing function of $n$. 

The core idea of our 2QCFA $A$ for the word problem $W_G$ is to scan the input word $w=w_1 \cdots w_n \in \Sigma^*$ and, after the partial word $w_1 \cdots w_t$ has been read, the quantum register of $A$ stores the group element $g_t:=\phi(w_1\cdots w_t) \in G$. On inputs of string length $n$, $g_t$ may vary over the entirety of $B_{G,S}(n)$. In order to store an arbitrary element of $B_{G,S}(n)$ such that it is (information theoretically) possible to perfectly discern the identity of that element, one requires $\log(f_{G,S}(n))$ (classical) bits. Moreover, by Holevo's theorem \cite{holevo1973bounds}, this same task requires $\log f_{G,S}(n)$ qubits.  

Therefore, we must first make clear why our approach, which encodes such an element using only a single qubit, does not violate Holevo's theorem. The key observation is that, while all $\log f_{G,S}(n)$ bits of information are truly stored in the single qubit, one is extremely limited in the manner in which that information may be accessed. In particular, this information may only be accessed by performing a quantum measurement, which only (probabilistically) indicates whether or not the currently stored value $g_t$ is equal to the identity element $1_G$; moreover, performing this quantum measurement completely destroys all information stored in this qubit. This extremely severe restriction on the manner in which the information content of a qubit may be accessed prevents one from reconstructing information stored within the qubit in a manner inconsistent with Holevo's theorem. On the other hand, this restriction is perfectly consistent with the manner in which $A$ operates when solving the word problem of $G$, and so it provides no impediment to using a single qubit to store information in a radically compressed way.

We next quantify the extent to which our constructions of 2QCFA compress information. For two monotone non-decreasing functions $f_1,f_2:\mathbb{R}_{\geq 0} \rightarrow \mathbb{R}_{\geq 0}$, we write $f_1 \prec f_2$ if there are constants $C_1,C_2 \in \mathbb{R}_{>0}$ such that, $f_1(x) \leq C_1 f_2(C_1 x+C_2)+C_2$, $\forall x \in \mathbb{R}_{\geq 0}$, and we write $f_1 \sim f_2$ if both $f_1 \prec f_2$ and $f_2 \prec f_1$. Note that while the exact value of $f_{G,S}(n)$ does depend on $S$, the asymptotic behavior does not, in that $f_{G,S} \sim f_{G,S'}$, for any other finite generating set $S'$ \cite[Proposition 6.2.4]{loh2017geometric}; therefore, we will simply write $f_G$ in place of $f_{G,S}$ when only the asymptotic behavior is relevant. We say $G$ is of \textit{polynomial growth} if $f_G \sim n^C$, for some $C \in \mathbb{R}_{\geq 0}$, and of \textit{exponential growth} if $f_G \sim C^n$, for some $C \in \mathbb{R}_{>0}$. By the famous Tits' alternative \cite{tits1972free}, every $G \in \mathcal{L}$ is either of polynomial or exponential growth; in particular,  $G \in \mathcal{L}$ has polynomial growth precisely when it is virtually nilpotent. 

In particular, any finitely generated virtually abelian group $G$ has polynomial growth; therefore, one requires $\log f_G(n) \sim \log(n)$ classical bits to unambiguously store an element of $B_{G,S}(n)$. By Theorem~\ref{thm:main:abelian}, for any such $G$, there is a single-qubit 2QCFA $A$ that recognizes $W_G$, with bounded error, in expected polynomial time. In particular, $A$ stores this arbitrary element of $B_{G,S}(n)$ using only a single qubit. More dramatically, by Theorem~\ref{thm:main:freeDirProd}, for any finitely generated virtually free group $G$, there is a single-qubit 2QCFA $A$ that recognizes $W_G$, with bounded error, in expected exponential time. Any such $G$ which is not virtually cyclic (i.e., any such $G$ that is neither finite nor virtually $\mathbb{Z}$) has exponential growth, which means that one requires $\log f_G(n) \sim n$ classical bits to unambiguously store an element of $B_{G,S}(n)$. Yet, $A$ still stores an arbitrary element of $B_{G,S}(n)$ using only one qubit. 

The above examples, and more generally all of the 2QCFA that we have constructed for various word problems, demonstrate the extreme sort of information compression that a 2QCFA is capable of performing. On the other hand, this extreme compression does not come without a cost, as it directly impacts the running time of our 2QCFA. Moreover, this cost \textit{cannot} be avoided, as we have proven a corresponding lower bound \cite{remscrim2019lower}.

We note that information compression of this form is by no means a new idea in quantum computing, as techniques like quantum fingerprinting \cite{buhrman2001quantum} and dense quantum coding \cite{ambainis2002dense} explicitly involve such compression, and, moreover, many quantum algorithms, including Shor's quantum factoring algorithm \cite{shor1994algorithms}, crucially rely on this sort of compression to achieve their apparent speedup relative to their classical counterparts. Nevertheless, both the original Ambainis and Watrous 2QCFA result \cite{ambainis2002two} and our approach push this idea down to the much weaker computational model of 2QCFA, and introduce techniques that might also be useful for more powerful quantum models.

\section*{Acknowledgments}

The author would like to express his sincere gratitude to Professor Michael Sipser for many years of mentorship and support, without which this work would not have been possible, as well as to thank Professor David Vogan for a very helpful conversation.

\bibliographystyle{plainurl}
\let\OLDthebibliography\thebibliography
\renewcommand\thebibliography[1]{
	\OLDthebibliography{#1}
	\setlength{\parskip}{0pt}
	\setlength{\itemsep}{0pt plus 0.3ex}
}
\bibliography{qfaLinGroupRef} 

\end{document}